%% file: main.tex
\documentclass[sigconf]{acmart}

\input{style.tex}

\AtBeginDocument{%
  \providecommand\BibTeX{{%
    \normalfont B\kern-0.5em{\scshape i\kern-0.25em b}\kern-0.8em\TeX}}}

\copyrightyear{2020} 
\acmYear{2020} 
\setcopyright{acmlicensed}\acmConference[AFT '20]{2nd ACM Conference on Advances in Financial Technologies}{October 21--23, 2020}{New York, NY, USA}
\acmBooktitle{2nd ACM Conference on Advances in Financial Technologies (AFT '20), October 21--23, 2020, New York, NY, USA}
\acmPrice{15.00}
\acmDOI{10.1145/3419614.3423264}
\acmISBN{978-1-4503-8139-0/20/10}

\usepackage{color}
 
\acmSubmissionID{aft20-p62}

\begin{document}

\title{Efficient MDP Analysis for Selfish-Mining in Blockchains} 
\author{Roi Bar-Zur}
\affiliation{%
  \institution{Technion}
}

\author{Ittay Eyal}
\affiliation{%
  \institution{Technion}
}

\author{Aviv Tamar}
\affiliation{%
  \institution{Technion}
}

\renewcommand{\shortauthors}{Bar-Zur, Eyal and Tamar}

\begin{abstract}
A \emph{proof of work} (\emph{PoW}) \emph{blockchain protocol} distributes rewards to its participants, called miners, according to their share of the total computational power. 
Sufficiently large miners can perform \emph{selfish mining}~-- deviate from the protocol to gain more than their fair share.
Such systems are thus secure if all miners are smaller than a \emph{threshold size} so their best response is following the protocol. 

To find the threshold, one has to identify the optimal strategy for miners of different sizes, i.e., solve a \emph{Markov Decision Process} (\emph{MDP}). 
However, because of the PoW \emph{difficulty adjustment} mechanism, the miners' utility is a non-linear ratio function. 
We therefore call this an \emph{Average Reward Ratio} (\emph{ARR}) MDP. 
Sapirshtein et al.\ were the first to solve ARR MDPs by solving a series of standard MDPs that converge to the ARR MDP solution. 

In this work, we present a novel technique for solving an ARR MDP by solving a single standard MDP. 
The crux of our approach is to augment the MDP such that it terminates randomly, within an expected number of rounds. 
We call this \emph{Probabilistic Termination Optimization} (\emph{PTO}), and the technique applies to any MDP whose utility is a ratio function. 
We bound the approximation error of PTO~-- it is inversely proportional to the expected number of rounds before termination, a parameter that we control. 
Empirically, PTO's complexity is an order of magnitude lower than the state of the art.

PTO can be easily applied to different blockchains. 
We use it to tighten the bound on the threshold for selfish mining in Ethereum. 
\end{abstract}



\maketitle

\section{Introduction}
Proof of Work~\cite{dwork1992pricing,jakobsson1999proofs} Blockchain Protocols secure about~80\% of the cryptocurrency market cap~\cite{coinmarketcap2020}, including Bitcoin~\cite{nakamoto2008bitcoin} and Ethereum~\cite{buterin2013ethereum}. 
The operators of blockchains, called \emph{miners}, aggregate user transactions in \emph{blocks}, and order the blocks by forming a block chain. 
These are \emph{decentralized} protocols, allowing anyone to join and add blocks by proving she has performed some amount of work~-- a computational task. 
We call the rate of work expended by each miner its \emph{mining power} and normalize the sum of all rates to one. 
An attacker can thus violate the system guarantees only if she performs more work than all other participants combined, i.e., if its mining power is larger than~50\%. 
With no central enforcement, the security of a PoW blockchain relies critically on the fact that the best response of each miner is to follow the prescribed protocol. 

To incentivize the miners to follow the protocol, it rewards them with cryptocurrency tokens for their efforts. 
Moreover, the rewards incentivize entities to join the system, become miners, and work for the rewards~\cite{tsabary2018gap,eghbali2019twelve}. 
To regulate block generation rate, the protocol automatically \emph{adjusts the difficulty}, i.e., the amount of work necessary to produce each block such that the average \emph{block rate} remains constant. 
Therefore, if all miners follow the protocol, they each receive in expectation a fraction of the total reward equal to their mining power. 

The security of a blockchain protocol is thus defined by the minimal (threshold) mining power above which a miner can increase her revenue by deviating from the prescribed protocol~\cite{eyalmajority, sapirshtein2016optimal, nayak2016stubborn, kwon2017selfish, pass2017fruitchains}. 
In order to calculate the security threshold, we would like to find the smallest miner size such that its optimal strategy is not the prescribed one. 
A natural approach is to model the system as a \emph{Markov Decision Process}. 
This is a Markov Process where in each state the agent (miner) chooses an action, probabilistically transitioning the system into a new state. 
The utility of the miners is their fraction of the reward ratio due to the difficulty adjustment, i.e., a non-linear utility function. 
We call this an \emph{Average Reward-Ratio MDP} (\emph{ARR-MDP}). 
Unfortunately, standard MDP solution techniques are not compatible with the non-linear utility in ARR-MDP, calling for new solution methods. 

To the best of our knowledge, Sapirshtein et al.~\cite{sapirshtein2016optimal} were the first to solve an ARR-MDP (\S\ref{section:related_work}). 
Their approach is to devise a parametrized linear utility function for the original MDP. 
They perform a binary search for the value of the parameter: 
For each value they solve the MDP, finding the optimal strategy for the utility function. 
The search converges towards a parameter for which the optimal strategy for the original ARR utility function coincides with the optimal strategy for the parametrized utility function. 

In this work, we propose a novel technique~(\S\ref{section:method}), called \emph{Probabilistic Termination Optimization} (\emph{PTO}), to solve  an ARR-MDP by transforming it to a conventional MDP, which we call PT-MDP. 
The form of the utility in ARR-MDP is a ratio between the miner's reward and a value determined by the protocol's adjustment scheme.
For example, in Bitcoin, the miner's reward is the number of her blocks divided by the total number of blocks mined by the whole network~\cite{eyalmajority, sapirshtein2016optimal, nayak2016stubborn}. 
Thus, ARR-MDP can be thought of as a repetitive process divided to epochs, where each epoch contributes a constant amount towards the difficulty adjustment.
For Bitcoin~-- in each epoch the entire network mines a certain number of blocks.
The epochs are all equivalent so optimizing a single epoch gives the best strategy for the infinite ARR-MDP.
In the heart of the transformation is the idea that it suffices to optimize for epochs with a set \emph{expected} contribution. 
This is done by constructing PT-MDP such that each step has a probability of terminating the process.
The probabilistic termination ensures PT-MDP has a chosen expected amount of contribution towards the difficulty adjustment.
For Bitcoin~-- PT-MDP represents an epoch in which the entire network mines~$x$ blocks in expectation. 

This relaxation simplifies PT-MDP to have a linear objective function and still provides a provably accurate approximation.
We prove (\S\ref{section:proof}) the agent's revenue for the optimal policy in the transformed MDP is a good approximation for the agent's revenue in the original MDP.
In order to prove this we first show ARR-MDP and PT-MDP are \emph{ergodic}.
This allows us to use classical results \cite{serfozo2009basics} for MDPs to bound the approximation error of the transformation.
We obtain a tight bound inversely proportional to the expected number of amount of contribution chosen.
This parameter provides the ability to obtain results with an arbitrarily small approximation error.

One key advantage of PTO is its performance~(\S\ref{section:performance})~-- by solving only a single MDP its complexity is a fraction of the state of the art. 
To evaluate we compare against the Optimal Selfish Mining Strategies in Bitcoin
~\cite{sapirshtein2016optimal} (OSM) optimizer of Sapirshtein et al. 
In PTO we use \emph{Policy Iteration} to optimize the MDP, while Sapirshtein et al.\ used \emph{relative value iteration} whose performance is inferior in this case. 
We therefore compare PTO against an improved version of OSM we denote \emph{PI OSM} which uses Policy iteration. 
We first show that PI OSM is significantly faster than OSM running on the same platform. 
Then, to compare against PTO, which runs on a different platform, we compare the number of linear system solutions~-- the main computational bottleneck, showing an order of magnitude reduction in PTO. 

Finally, we demonstrate the ease of application of PTO by finding optimal strategies in the Ethereum blockchain(\S\ref{section:ethereum}). 
Ethereum's reward mechanism is more complicated than Bitcoin's. 
It grants partial rewards to blocks not in the chain, and counts them towards the difficulty adjustment. 
We use PTO to estimate the threshold size for a rational miner not to deviate, reducing it from 0.26~\cite{feng2019selfish, grunspan2019selfish} to 0.2468.

We released our code of PTO along with the Bitcoin and Ethereum models in a public GitHub repository\footnote{https://github.com/roibarzur/pto-selfish-mining}.


    \section{Related Work}\label{section:related_work}

Eyal and Sirer~\cite{eyalmajority} were the first to show that the Bitcoin protocol is not incentive compatible.
They have demonstrated and analyzed a strategy (called Selfish Mining or \emph{SM1}) which can be used by miners and yields higher rewards than acting honestly.
Selfish mining involves withholding newly minted blocks and violating the longest chain rule.
The existence of such a strategy rules out the honest protocol as a Nash equilibrium.

Their analysis uses a \emph{Markov chain} and gives a closed form formula of the revenue a miner can get if she uses SM1 and the minimal threshold for the relative computational power a miner needs such that using SM1 is superior to following the protocol.
However, the analysis does not guarantee that a miner with a relative power lower than the threshold cannot deviate in some other way and profit.
This means that the SM1 threshold is an upper bound on the security threshold of the protocol (the minimal relative power required in order to be able to deviate and profit).
Overall, this analysis implies that under the reasonable assumption of $\gamma = 0.5$ the security threshold of Bitcoin is 0.25 at most.

Nayak et al.~\cite{nayak2016stubborn} design improved selfish mining strategies and use numeric simulations to quantify their revenues for different parameters.
By doing so, they prove SM1 is not optimal and demonstrate that for varying parameters, different strategies are preferred.

Sapirshtein et al.~\cite{sapirshtein2016optimal} design a method to find the optimal selfish mining strategy for given parameters. They model the Bitcoin protocol as a Markov Decision Process with a non-linear reward criterion, and propose a solution method for approximately solving the MDP, by performing a binary search and in each stage solving a standard MDP.
Their analysis yields an accurate approximation of the security threshold of Bitcoin. 
For example, they find that for some parameters, the threshold is as low as~0.2321. 

In contrast, PTO allows using a single MDP solution to approximate the optimal strategy, reaching an arbitrarily accurate solution with an order of magnitude lower complexity. 
As PTO is more efficient, it is also easier to generalize it to more complex blockchain protocols.

Wang et al.~ \cite{wang2019blockchain} devise a different approach to overcome the non-linear reward criterion.
They design a generalization of Q-learning, a reinforcement learning algorithm suited to maximize the non-linear reward.
They name their new algorithm multi-dimensional RL.
Similarly to Q-learning, their approach is based on Monte Carlo simulations and is model-free, meaning it does not make use of available information such as the transition probabilities of the MDP.
Although this might be easier to apply to different protocols, it introduces noise and slows down the convergence speed significantly. 
In addition, they obtained empirical evidence that the new algorithm manages to converge to the optimum.
However, they do not provide a theoretical foundation to their new approach 
and it remains an open question whether it can generalize to other blockchain models as well.

Other work generalizes the SM1 strategy to Ethereum and analyzes its revenue in order to get an upper bound on the security threshold of Ethereum. 
Ritz and Zugenmaier~\cite{ritz2018impact} used a Monte-Carlo simulation in order to simulate their generalized SM1 strategy and assess its revenue.
Niu and Feng~\cite{feng2019selfish} also generalize the SM1 strategy to Ethereum and use a theoretical analysis by a Markov chain to calculate the revenue of the strategy.
Grunspan and P\'erez-Marco~\cite{grunspan2019selfish} use a theoretical analysis involving combinatorics for two versions of generalized SM1 strategies . 
Nevertheless, all of these works analyze specific strategies and therefore can only strive to obtain an upper bound of the security threshold.
We use PTO to estimate a tight bound using the optimal strategy. 

A new technique called SquirRL~\cite{hou2019squirrl} analyzes blockchain protocols using deep reinforcement learning. 
It is based on the iterative MDP solution method of OSM, and uses deep-RL to find an approximately optimal policy.
SquirRL was used for Bitcoin and Ethereum and obtained a better upper bound for the security threshold of Ethereum.
By observing their results for Bitcoin, they deduce that their method typically finds solutions with revenues approximately~1-2\% lower than the optimal revenue found by OSM.
Due to the neural network approximations in deep RL, the threshold of SquirRL is only guaranteed to be an upper bound, and deriving meaningful error bounds for their results is not possible. 
In contrast, PTO converges the optimum with a bound on the approximation error. 


    \section{Preliminaries} \label{section:prelim} 


        \subsection{Markov Decision Processes} \label{section:prelim:mdp} 

Our method is based on the theory of Markov decision processes (MDPs), which we now review. 
An MDP describes a controlled stochastic process in discrete time~\cite{white2001markov}, where at time~$t$ an \emph{agent} observes an environment at \emph{state}~$s_t\in \fancyS$, takes an \emph{action}~$a_t\in \fancyA$, and subsequently the environment transitions stochastically to a new state~$s' \in \fancyS$, while the agent is awarded some \emph{reward}~$R_t = R_a(s, s')$.
The transitions between states are Markovian, and denoted by:
\begin{displaymath}
\Pr{s_{t+1} = s' | s_t = s, a_t = a} =  P_a(s, s') \,\, .
\end{displaymath}


The agent's goal is to choose actions that maximize its long-term rewards, defined by the MDP \emph{objective function}. 
Three objectives that have been extensively investigated in the literature are 
the \emph{discounted reward}, the \emph{average reward}, and the \emph{stochastic shortest path} (\emph{SSP})~\cite{bertsekas1995dynamic}.
For all of these cases, an optimal decision making policy (i.e., an action strategy that maximizes the objective) can be represented as a Markov \emph{policy}~-- a deterministic mapping from the state space to the action space. 
Furthermore, several algorithms for finding an optimal policy are known. 
An MDP with a specific policy~$\pi$ induces a Markov chain over the states visited by the policy. 
In the sequel, the notation $\Epolicy{\cdot}$ denotes an expectation over states in the Markov chain induced by the policy $\pi$. 


\subsubsection{Discounted reward}
In this case the reward at time t is discounted by~$\beta^t$ for some~$0 < \beta \leq 1$, and the objective is:
\begin{displaymath}
R(\pi) = \Epolicy{\lim\limits_{T \to \infty} \sum\limits_{t=0}^T \beta^t R_t} \,\, .
\end{displaymath}

\subsubsection{Average reward}
In this case, the objective is the average reward per step across an infinite horizon:
\begin{displaymath}
R(\pi) = \Epolicy{\lim\limits_{T \to \infty} \frac{1}{T} \sum\limits_{t=0}^T R_t} \,\, .
\end{displaymath}
This objective is used in OSM.

\subsubsection{Stochastic shortest path}
This is a similar case to discounted reward with~$\beta=1$.
However, in this case the MDP is assumed to have a terminal absorbing state which is guaranteed to be reached for every policy.
Once this state is reached at step~$T$ the process always stays in this state and no further reward is obtained. 
Thus,~$T$ is a random stopping time, and the objective function is of the form:
\begin{displaymath}
R(\pi) = \Epolicy{\sum\limits_{t=0}^T R_t} \,\, .
\end{displaymath}
Our method PTO will make use of this objective.

Exact algorithms for solving MDPs with the objective criteria above include \emph{value iteration} and \emph{policy iteration}, which are guaranteed to converge to an optimal policy~\cite{bertsekas1995dynamic}. When the state space is too large for exact methods, reinforcement learning algorithms such as Q-learning~\cite{bertsekas1995dynamic,hou2019squirrl} provide an approximate solution.

\subsection{Markov Chains}
We now recapitulate several classical Markov chain definitions and results~\cite{serfozo2009basics}. 
Henceforth,~$Y_n$ denotes a Markov chain in general, and when applicable also denotes a random variable of the state of the Markov chain at time~$t$.
We let $\fancyS$ denote the state space of the Markov chain, and~$P$ denotes its transition matrix.

The \emph{hitting time} of a state is a random variable of the number of steps it takes to return to the state when starting at said state.
Formally, the hitting time of a state~$i$ by a process~$Y_n$ is defined by:
\begin{displaymath}
\tau_i = \min \{ t \geq 2 : Y_1 = Y_t = i \} \,\, .
\end{displaymath}
A \emph{period} of a state is the greatest common divisor of all possible values for its hitting time.
A state is called \emph{aperiodic} if its period is equal to 1.

A \emph{recurrent} state is a state for which the chain returns to with probability 1 assuming it starts in said state.
A \emph{transient} state is a state which is not recurrent.
A \emph{positive recurrent} state is a recurrent state for which the expected hitting time is finite.

An \emph{irreducible} chain is a chain in which for every pair of states~$i$ and~$j$ there is a chance to transition from~$i$ to~$j$ after any number of steps. 

The following lemma is a classical result by Serfozo~\cite{serfozo2009basics}.
\begin{lemma} \label{commumication_class_property_lemma}
In an irreducible chain, if one state is positive recurrent and/or aperiodic then all its states are positive recurrent and/or aperiodic.
In addition, the chain is called positive recurrent and/or aperiodic.
\end{lemma}

A Markov chain which is irreducible, and its states are positive recurrent and aperiodic is called \emph{ergodic}.

A probability measure $\mu$ on $\fancyS$ is a \emph{stationary distribution} for the Markov Chain~$Y_n$ (or for P) if for all $i \in \fancyS$ it holds that:
\begin{displaymath}
\mu_i = \sum\limits_{j \in \fancyS} \mu_j p_{ji} \,\, .
\end{displaymath}

An ergodic Markov chain alway has a stationary distribution.


    \section{Method} \label{section:method} 

We describe our method for computing security bounds for blockchain protocols based on the MDP model. 
We begin with a generalization~(\S\ref{section:method:generalization}) of the Bitcoin blockchain MDP model proposed by Sapirshtein et al.~\cite{sapirshtein2016optimal}, detail model assumptions on the MDP~(\S\ref{section:method:assumptions}) and then present our algorithm, PTO~(\S\ref{section:method:PTO}), for solving the MDP. 


        \subsection{An MDP Model for PoW Blockchain} \label{section:method:generalization} 

We consider an MDP with a finite state space and a finite action space modeling a general PoW blockchain protocol with a difficulty adjustment mechanism. The states, actions, and transitions in the MDP depend on the particular protocol, and in Sections~\ref{section:performance:model} and \ref{section:ethereum:model} we detail them for the cases of Bitcoin and Ethereum. The development in this section, however, focuses on the reward in the MDP, and is not specific to a particular protocol.

The agent is a rational miner who wants to maximize her reward per unit of time.
We assume all of the other miners in the blockchain act as prescribed.
The difficulty adjustment scheme slows down the mining so the rational miner reward is divided by a factor determined by the difficulty of mining.
We call this factor the \emph{difficulty contribution}. 
In Bitcoin for example, the reward is the number of blocks the miner appends to the main chain and the difficulty contribution is the number of all blocks appended \cite{eyalmajority, sapirshtein2016optimal}.
In Ethereum, on the other hand, the reward of a rational miner is the sum of her rewards from regular blocks, uncles and nephew rewards, and the difficulty contribution is the sum of of blocks added to the main chain and uncle blocks~\cite{grunspan2019selfish}.\footnote{Ethereum introduces the concept of \emph{uncle} blocks. An uncle block is a block which is not in the main chain but is a direct descendant of another block in the main chain. Apart from the regular block reward, Ethereum gives additional rewards to blocks who reference uncle blocks (called \emph{nephew rewards}) and to miners of referenced uncle blocks (called \emph{uncle rewards}). This is in order to compensate miners who find blocks which end up out of the main chain due to network latency. In order to deter miners from intentionally creating blocks meant to be uncles, Ethereum counts both regular blocks and uncle blocks for its difficulty adjustment \cite{ritz2018impact}.}
In the following, we consider a general MDP model where the agent needs to balance the rewards with a  difficulty contribution.

In our MDP, at each time step $t$, the agent obtains both a scalar reward $R_t$~-- the reward of the rational miner in step~$t$ and a scalar~$D_t\geq 0$~-- the difficulty contribution of the entire network at that time step.
Under the difficulty adjustment scheme~\cite{eyalmajority}, we model the objective function of the rational miner in the general form:
\begin{displaymath}
\rev = \E{\lim\limits_{T\to\infty} \frac{\sum\limits_{t=1}^T R_t}{\sum\limits_{t=1}^T D_t}} \,\, .
\end{displaymath}

Note that this objective function is different from the standard objective functions for MDPs, and  henceforth we denote such MDPs as \emph{average reward-ratio MDPs} (ARR-MDP). In particular, since the objective is not a linear function of the reward, standard MDP solution methods are not applicable.
Our main contribution in this work is an algorithm for approximately solving ARR-MDPs. We note that while our focus is on the blockchain domain, our algorithm and analysis apply to any ARR-MDP, which can potentially be used in other domains.

We note that the model in \cite{sapirshtein2016optimal} is in fact an ARR-MDP, where $D_t$ is defined as the sum of blocks mined by the agent and the rest of the network. Thus, the analysis of OSM applies to our setting as well. However, as we will show later, our approach leads to much more efficient ARR-MDP algorithms.


        \subsection{Model Assumptions} \label{section:method:assumptions}

We next describe several assumptions on the ARR-MDP that apply in the blockchain setting, and will allow us to develop our efficient solution method. We first bound the reward and difficulty contributions.
\begin{assumption}\label{assump_r_t_bounded}
There is a constant $\rmax \in \R^+$ such that for any policy~$\pi$, it holds that:
\begin{displaymath}
\forall_t : \abs{R_t} \leq \rmax \,\, .
\end{displaymath}
\end{assumption}
\begin{assumption}\label{assump_d_t_bounded}
There is a constant $\dmax \in \R^+$ such that for any policy~$\pi$, it holds that:
\begin{displaymath}
\forall_t : 0 \leq D_t \leq \dmax \,\, .
\end{displaymath}
\end{assumption}
For bitcoin, the reward and the difficulty contribution symbolize numbers of blocks added at every step,
and are bounded by the length of the longest possible fork in the network, which is reasonably bounded~\cite{sapirshtein2016optimal}.
A similar reasoning also works for Ethereum.
While not necessarily relevant to the blockchain setting, we remark that negative bounded rewards are allowed in our formulation.

\begin{assumption}\label{assump_d_t_not_zero}
For any policy~$\pi$, the average difficulty contribution in $\orgmdp$ is lower bounded by some constant~$\varepsilon > 0$.
Formally:
\begin{displaymath}
\Epolicy{\lim\limits_{T\to\infty} \frac{1}{T}{\sum\limits_{t=1}^T D_t}} > \varepsilon \,\, .
\end{displaymath}
\end{assumption}
Assumption \ref{assump_d_t_not_zero} holds in any PoW blockchain protocol, as regardless of the agents' actions, the rest of the network continues mining. Thus, it is reasonable to assume that the agents' actions do not halt the blockchain.

\begin{assumption}\label{assump_positive_reccurent_state}
There is a state in $\orgmdp$ that is positive recurrent for any policy~$\pi$.
Denote this state as the initial state or $\sinit$.
\end{assumption}
Assumption \ref{assump_positive_reccurent_state} holds for Bitcoin and Ethereum 
if we assume the miner does not have more than~50\% of the mining power, which is reasonable since otherwise the blockchain is already compromised. Thus, the honest miners always catch up with the rational miner, and since she cannot keep waiting forever, she must sync with the rest of the network at some point and get back to the initial state.
A similar assumption also appears in the analysis of OSM with the same reasoning.

In general, different protocols may require a model in which syncing does not lead to the same state, so this reasoning might not always fit.
For example, a protocol which remembers all history will have a different state after syncing as it captures what happened since the last sync.
However, in order to use exact MDP solving algorithms, the state space has to be finite and when modeling such protocols one has to truncate the state space to be able to use those algorithms.
A byproduct of this is that we get a recurrent state since all finite space MDPs have a recurrent state.
So overall, this assumption does not pose additional restrictions.


        \subsection{Probabilistic Termination Optimization} \label{section:method:PTO}

We now present our method for solving ARR-MDPs, which we term Probabilistic Termination Optimization (PTO). 
Our main observation is that we can construct an auxiliary MDP that approximates the ARR-MDP to arbitrary precision. 
This auxiliary MDP is in a standard SSP form, and can therefore be optimized using standard methods. 
In the sequel, we will explain the intuition behind our approach, show how to define the auxiliary MDP, how to solve it efficiently, and how to bound the precision of our approximation of the true ARR-MDP.

            \subsubsection{Intuition}

First, we remark that the agent's objective is to maximize:
\begin{displaymath}
\rev = \E{\lim\limits_{T\to\infty} \frac{\sum\limits_{t=1}^T R_t}{\sum\limits_{t=1}^T D_t}} \,\, .
\end{displaymath}
Let us go over some potential solutions to overcome the non-linear form of the agent's objective function.

If we tried substituting the limit with a fixed T we would get:
\begin{displaymath}
\rev = \E{\frac{\sum\limits_{t=1}^T R_t}{\sum\limits_{t=1}^T D_t}} \,\, .
\end{displaymath}
This is equivalent to terminating after a fixed number of steps and for large T would give a close approximation of the actual limit.
However, this still leaves the objective function in a non-linear form.

Another option could be to construct an MDP that terminates once $\sum\limits_{t=1}^T D_t$ is equal to some parameter $H$.
By doing this, we would fix the denominator and be left with a linear function to maximize:
\begin{displaymath}
\rev = \E{\frac{1}{H} \sum\limits_{t=1}^T R_t} \,\, .
\end{displaymath}
For large values of~$H$, the process runs until a very large~$T$ and similarly to the previous option, we know that we would get a good approximation.
However, the state space in this MDP would have to keep a memory of the accumulated $D_t$ in order to know when to terminate. 
This is undesirable since keeping track of time implies a prohibitive increase in the state space. 

Our proposed solution is to introduce a memoryless termination probability such that in expectation $\sum\limits_{t=1}^T D_t$ is $H$  when the process terminates.
From the first option we know that a large $H$ we would be close to the limit.
Like the second solution, we get rid of the denominator. 
Although now the equality holds only in expectation, we will prove later that this still provides a good approximation.
In addition, since the termination probability is memoryless 
we do not have to keep track of the accumulated $D_t$ so we can use the original state space.

In order to create a memoryless termination probability we utilize independent coin tosses.
For every one unit of accumulated~$D_t$, the MDP tosses a coin:
Terminate the process with probability~$\frac{1}{H}$ or continue with probability~$1 - \frac{1}{H}$.
Intuitively, since each unit of accumulated $D_t$ causes termination with probability $\frac{1}{H}$ the accumulated $D_t$ resembles a geometric distribution and therefore in expectation would be $H$ once termination occurs. 

            \subsubsection{The Auxiliary MDP}

For an ARR-MDP we denote its auxiliary MDP by~$\auxmdp$.
$\auxmdp$ has the same state space as~$\orgmdp$, with an additional terminal state with zero reward.
PT-MDP is parametrized by some chosen parameter $H$. At every time step~$t$, the agent in~$\auxmdp$ has a probability of~$1 - (1 - \frac{1}{H})^{D_t}$, to transition to the terminal state.
If termination has not occurred, then the transition occurs as in $\orgmdp$.
Intuitively, the process continues only if all independent coin tosses (there are $D_t$ of them) indicate to continue.

Let $\termaux(H)$ be a random variable indicating the step in which the process moves to the terminal state.
Then the objective function of~$\auxmdp$ is the stochastic shortest path criterion (multiplied by a constant):
\begin{displaymath}
\revaux(H) = \E{\frac{1}{H} \sum\limits_{t=0}^{\termaux(H)} R_t} \,\, .
\end{displaymath}

Note that the probability of termination depends on $D_t$, and the higher $D_t$ is, the higher the chance to terminate.
Intuitively, this encourages the agent to strike a balance between the reward and the difficulty contribution.
We call $H$ the \emph{expected horizon} since, as we will formally show later, it corresponds to the expected total difficulty contribution until termination.



            \subsubsection{Solving the MDPs} 

To solve $\auxmdp$, we can use any standard SSP algorithm, such as value iteration or policy iteration. Assumption \ref{assump_d_t_not_zero} ensures that there will always be a positive chance of termination in $\auxmdp$, a condition that guarantees that running any of these algorithms will result in an approximately optimal policy $\pi$~\cite{bertsekas1995dynamic}. 
This policy can be directly applied to the $\orgmdp$, as the states and actions in $\auxmdp$ and $\orgmdp$ are the same (up to the terminal state, which is not relevant for the policy).
We emphasize that our goal is to analyze the revenue in the original $\orgmdp$.
To do this, we calculate the steady-state distribution of the Markov chain induced by the policy $\pi$ in $\orgmdp$, and based on this distribution calculate the expected reward 
\begin{displaymath}
\E{\lim\limits_{T\to\infty} \frac{1}{T}{\sum\limits_{t=1}^T R_t}} \,\, ,
\end{displaymath}
and expected difficulty contribution
\begin{displaymath}
\E{\lim\limits_{T\to\infty} \frac{1}{T}{\sum\limits_{t=1}^T D_t}} \,\, .
\end{displaymath}
The revenue is then the ratio between these two values \cite{sapirshtein2016optimal}.


    \section{Proof of Optimality}\label{section:proof} 

We next prove that PTO converges to the optimal solution of ARR-MDP. 
We bound the difference in revenue between~$\orgmdp$ and~$\auxmdp$ and show that it can be made arbitrarily small. 
%
%
%
%

To specify our result we first need some notation. 
Denote the revenue of the rational miner in $\orgmdp$ under policy$~\pi$ by 
$$ 
\revorg ^\pi \triangleq \Epolicy{\lim\limits_{T\to\infty} \frac{\sum\limits_{t=1}^T R_t}{\sum\limits_{t=1}^T D_t}} \,\, . 
$$ 
%
At every step $t$, $X_t$ is a random variable indicating whether the process $\auxmdp$ terminated at this step, that is, 
$$ 
X_t \triangleq
\begin{cases}
1 & \text{w.p } \left( 1 -\frac{1}{H} \right)^{D_t} \>\>\>\>\>\>\>\> \text{(continue)}\\
0 & \text{w.p } 1 - \left( 1 -\frac{1}{H} \right)^{D_t} \>\> \text{(stop)}
\end{cases} \,\, .
$$
%
Denote the first time at which $\auxmdp$ terminates by
$$ 
\termaux(H) \triangleq \arg \min\limits_t \{ X_t = 0 \} \,\, .
$$ 

Denote the revenue of the rational miner in $\auxmdp$ under policy$~\pi$ as by
$$
\revaux^\pi(H) \triangleq \Epolicy{\frac{1}{H}\sum\limits_{t=1}^{\termaux(H)} R_t} \,\, .
$$

Our main theorem bounds the difference in revenue for the same policy in~$\orgmdp$ and~$\auxmdp$, showing it is linear in~$\frac{1}{H}$. Formally, 
\begin{theorem}\label{main_theorem}
Under Assumptions \ref{assump_r_t_bounded}--\ref{assump_positive_reccurent_state}, for any policy~$\pi$, the approximation error of PTO is bounded, 
\begin{displaymath}
\abs{\revorg^\pi -  \revaux^\pi(H)} = \bigO{\frac{1}{H}} \,\, .
\end{displaymath}
\end{theorem}

It follows directly that when optimizing $\auxmdp$, as $H \to \infty$, the revenue of the optimal policy of $\auxmdp$ converges to the revenue of the optimal policy in $\orgmdp$.
\begin{corollary}
Under Assumptions \ref{assump_r_t_bounded}--\ref{assump_positive_reccurent_state}, for any policy~$\pi$, it holds that:
\begin{displaymath}
\lim\limits_{H \to \infty} \max\limits_\pi \revaux^\pi(H) = \max\limits_\pi \revorg^\pi \,\, .
\end{displaymath}
\end{corollary}

In order to prove Theorem~\ref{main_theorem}, we begin with necessary notation~(\S\ref{section:proof:notation}) and then show that $\orgmdp$ and $\auxmdp$ are ergodic~(\S\ref{section:proof:ergodic}). 
Afterwards, we present the expected horizon lemma~(\S\ref{section:proof:horizon}), which shows that when $\auxmdp$ terminates, the difficulty contribution is approximately~$H$.
We then bound the difference in average reward and average difficulty contribution between~$\revorg$ and~$\revaux$ (\S\ref{section:proof:average_difference}).
Later, we simplify the expressions for~$\revorg$ and~$\revaux$~(\S\ref{section:proof:simplify}), and finally conclude by proving Theorem \ref{main_theorem}~(\S\ref{section:proof:theorem}).


        \subsection{Notation} \label{section:proof:notation}
        
From now on until the end of the proof, when considering either $\orgmdp$ or $\auxmdp$, we will assume a fixed policy~$\pi$.
This means that the MDPs are reduced to Markov chains as any action the agent takes is determined by~$\pi$.

Denote by~$\sinitaux$ the terminal state of $\auxmdp$.
Denote by~$\Porg(i,j)$ the chance to transition to state~$j$ when $\orgmdp$ is in state~$i$, so $\Porg$ is the transition matrix of $\orgmdp$.
Denote by~$\Paux(i,j)$ the chance to transition to state~$j$ when $\auxmdp$ with parameter~$H$ is in state~$i$, so~$\Paux$ is the transition matrix of~$\auxmdp$.
Also denote by~$\Rmdp(i,j)$ the reward in $\orgmdp$ when transitioning from state~$i$ to state~$j$ and $\Dmdp(i,j)$ as the difficulty contribution of the same transition.
Note that $\Rmdp(i,j)$ and $\Dmdp(i,j)$ are equivalent in $\orgmdp$ and $\auxmdp$ since only the transition probabilities were changed.

Denote the expected reward after state~$i$ in~$\orgmdp$ by 
\begin{displaymath}
\hat{R}(i) \triangleq \sum\limits_{j \in \fancyS} \Rmdp(i,j) \Porg(i,j) \,\, .
\end{displaymath}

Denote the expected reward after state~$i$ in~$\auxmdp$ by
\begin{displaymath}
\hat{R}'_H(i) \triangleq \sum\limits_{j \in \fancyS} \Rmdp(i,j) \Paux(i,j) \,\, .
\end{displaymath}

Denote the expected difficulty contribution after state $i$ in $\orgmdp$ by
\begin{displaymath}
\hat{D}(i) \triangleq \sum\limits_{j \in \fancyS} \Dmdp(i,j) \Porg(i,j) \,\, .
\end{displaymath}

Denote the expected difficulty contribution after state~$i$ in~$\auxmdp$ by 
\begin{displaymath}
\hat{D}'_H(i) \triangleq \sum\limits_{j \in \fancyS} \Dmdp(i,j) \Paux(i,j) \,\, .
\end{displaymath}

We will also use matrix and vector notations for all functions of the state (e.g.~$P$, $\hat{R}$, etc.) while regarding all vectors as column vectors and denote the dot product as $\inner{\cdot}{\cdot}$ when summing over the state space.

We will next see that $\orgmdp$ and $\auxmdp$ are ergodic and this implies they have stationary distributions. We denote the stationary distributions for $\orgmdp$ and $\auxmdp$ as ~$\piorg$ and~$\piaux$, respectively.

A classical result based on the MDP ergodicity~\cite{serfozo2009basics} can be used to simplify the average reward/difficulty contribution per step in both $\orgmdp$ and $\auxmdp$.
The average reward per step in $\orgmdp$ is $\inner{\hat{R}}{\piorg}$, 
the average reward per step in $\auxmdp$ is $\inner{\hat{R}'_H}{\piaux}$,
the average difficulty contribution per step in $\orgmdp$ is $\inner{\hat{D}}{\piorg}$,
and the average difficulty contribution per step in $\auxmdp$ is $\inner{\hat{D}'_H}{\piaux}$.


        \subsection{Proof of Ergodicity} \label{section:proof:ergodic}
        
In order to use classical results regarding ergodic Markov chains, we prove ~$\orgmdp$ and $\auxmdp$ are ergodic.
\begin{lemma}\label{erogidic_thm}
It holds that $\orgmdp$ and $\auxmdp$ are ergodic.
\end{lemma}
\begin{proof}
In this proof, to use ergodicity results, we will assume that once the agent in $\auxmdp$ terminates and enters~$\sinitaux$, the process next transitions to~$\sinit$ and restarts.
This transition happens w.p~1 and does incur any reward or contribution to the difficulty.
This trick does not change our results (which only concern the rewards until the first termination), and is only for mathematical convenience.

First, we can say that $\orgmdp$ and $\auxmdp$ are irreducible w.l.o.g because we can ignore all the states that are unreachable from~$\sinit$ and any transient states before it.
This does not change the reward criteria since transient states do not affect average rewards as they only occur a finite amount of times and unreachable states do not occur at all.

From Assumption~\ref{assump_positive_reccurent_state}, $\sinit$ is positive recurrent in $\orgmdp$.
The trick to restart $\auxmdp$, combined with Assumption~\ref{assump_d_t_not_zero} imply that~$\sinit$ is positive recurrent in $\auxmdp$ as well.

In addition, We can say that both $\orgmdp$ and $\auxmdp$ are aperiodic w.l.o.g because we can take the positive recurrent state and change it to have a chance to transition to itself with no reward and no difficulty contribution.
This does not change $\revorg$ or $\revaux$ at all since this just means that the game may halt for a few steps and then carry on normally.
But, this change ensures the state has a period of 1.

Since both MDPs are irreducible and there is an aperiodic positive recurrent state, Lemma~\ref{commumication_class_property_lemma} gives that all the states are aperiodic and positive recurrent, showing that the MDPs are irreducible and all their states are aperiodic positive recurrent, and therefore ergodic.
\end{proof}


        \subsection{Expected Horizon Lemma} \label{section:proof:horizon} 

This following lemma shows that for~$H \gg \dmax$, the expected sum of the difficulty contribution from the start of the process until termination is a close approximation of~$H$, explaining why we call~$H$ the expected horizon. 
This result also explains the intuition behind the choice of $\revaux$~-- instead of dividing by the actual difficulty contribution, we divide by the expected difficulty contribution.

\begin{lemma}\label{diff_cont_lemma}
The expected total contribution to the difficulty when $\auxmdp$ terminates is equal to~$H$ up to~$\dmax + 1$.
Formally, it holds that:
\begin{displaymath}
\abs{\E{\sum\limits_{t=1}^{\termaux(H)} D_t} - H} \leq \dmax + 1 \, \, .
\end{displaymath}
\end{lemma}

We defer the proof to Appendix~\ref{section:proofs_appendix:horizon}.


\subsection{Bounding the Average Difference} \label{section:proof:average_difference} 

The next lemma bounds the difference of the average reward/difficulty contribution per step in $\orgmdp$ and $\auxmdp$.

\begin{lemma}\label{lemma_expected_reward_diff_cont_equal}
The expected reward per step of $\orgmdp$ and $\auxmdp$ are equal up to $\bigO{\frac{1}{H}}$.
Formally, it holds that:
\begin{displaymath}
\abs{\inner{\hat{R}}{\piorg} - \inner{\hat{R}'_H}{\piaux}} = \bigO{\frac{1}{H}} \, \, .
\end{displaymath}
Furthermore, The expected difficulty contribution per step of $\orgmdp$ and $\auxmdp$ are also equal up to $\bigO{\frac{1}{H}}$.
Formally, it holds that:
\begin{displaymath}
\abs{\inner{\hat{D}}{\piorg} - \inner{\hat{D}'_H}{\piaux}} = \bigO{\frac{1}{H}} \, \, .
\end{displaymath}
\end{lemma}

Here we present a short lemma, which uses the previous the lower bound the average difficult contribution per step of both ARR-MDP and PT-MDP.

\begin{lemma}\label{lemma_d_t_org_aux_not_zero}
The average difficulty contribution per step in $\orgmdp$ and $\auxmdp$ is more than some constant $\varepsilon$ > 0.
Formally, it holds that
\begin{displaymath}
\inner{\hat{D}}{\piorg} > \varepsilon \, \, ,
\end{displaymath}
and
\begin{displaymath}
\inner{\hat{D}'_H}{\piaux} > \varepsilon \, \, .
\end{displaymath}
\end{lemma}

The proof are deferred to Appendix~\ref{section:proofs_appendix:auxiliary_lemmas}.


        \subsection{Simplifying the Revenue} \label{section:proof:simplify}
        
The following lemmas are the first direct steps towards bounding the approximation error.

\begin{lemma}\label{simplified_revorg}
The revenue in $\orgmdp$ is equal to the average expected reward per step divided by the average expected difficulty contribution per step.
Formally, it holds that:
\begin{displaymath}
\revorg = \frac{\inner{\hat{R}}{\piorg}}{\inner{\hat{D}}{\piorg}} \, \, .
\end{displaymath}
\end{lemma}

\begin{lemma}\label{simplified_revaux}
The revenue in $\auxmdp$ is equal to the average expected reward per step divided by the average expected difficulty contribution per step up to $\bigO{\frac{1}{H}}$.
Formally, it holds that:
\begin{displaymath}
\abs{\revaux - \frac{\inner{\hat{R}'}{\piaux}}{\inner{\hat{D}'}{\piaux}}} = \bigO{\frac{1}{H}} \, \, .
\end{displaymath}
\end{lemma}

The proofs are deferred to Appendix~\ref{section:proofs_appendix:simplify}.


        \subsection{Completing the Proof} \label{section:proof:theorem} 
        
We are now ready to prove our main result, Theorem \ref{main_theorem}:
\begin{proof}
We first bound the difference of the simplified forms in Lemmas~\ref{simplified_revorg} and ~\ref{simplified_revaux}.
\begin{multline*}
\left| \frac{\inner{\hat{R}'_H}{\piaux}}{\inner{\hat{D}'_H}{\piaux}} - \frac{\inner{\hat{R}}{\piorg}}{\inner{\hat{D}}{\piorg}} \right| = \left| \frac{\inner{\hat{R}'_H}{\piaux} \cdot \inner{\hat{D}}{\piorg} - \inner{\hat{R}}{\piorg} \cdot \inner{\hat{D}'_H}{\piaux}}{\inner{\hat{D}'_H}{\piaux} \cdot \inner{\hat{D}}{\piorg}} \right|
\end{multline*}
Then, we use Lemma~\ref{lemma_d_t_org_aux_not_zero} to lower bound the denominators.
\begin{multline*}
\left| \frac{\inner{\hat{R}'_H}{\piaux}}{\inner{\hat{D}'_H}{\piaux}} - \frac{\inner{\hat{R}}{\piorg}}{\inner{\hat{D}}{\piorg}} \right| =\\
= \left| \frac{\inner{\hat{R}'_H}{\piaux} \cdot \inner{\hat{D}}{\piorg} - \inner{\hat{R}}{\piorg} \cdot \inner{\hat{D}'_H}{\piaux}}{\inner{\hat{D}'_H}{\piaux} \cdot \inner{\hat{D}}{\piorg}} \right| \leq\\
\leq \varepsilon^{-2} \left| \inner{\hat{R}'_H}{\piaux} \cdot \inner{\hat{D}}{\piorg} - \inner{\hat{R}}{\piorg} \cdot \inner{\hat{D}'_H}{\piaux} \right| \leq\\
\leq \varepsilon^{-2} \big| \inner{\hat{R}'_H}{\piaux} \cdot \inner{\hat{D}}{\piorg} - \inner{\hat{R}}{\piorg} \cdot \inner{\hat{D}}{\piorg} +\\
+ \inner{\hat{R}}{\piorg} \cdot \inner{\hat{D}}{\piorg} - \inner{\hat{R}}{\piorg} \cdot \inner{\hat{D}'_H}{\piaux} \big| \leq\\
\leq \varepsilon^{-2} \left| \left( \inner{\hat{R}'_H}{\piaux} - \inner{\hat{R}}{\piorg} \right) \cdot \inner{\hat{D}}{\piorg} \right| +\\
+ \left| \inner{\hat{R}}{\piorg} \cdot \left( \inner{\hat{D}}{\piorg} - \inner{\hat{D}'_H}{\piaux} \right) \right|
\end{multline*}
We now use the fact that stationary distributions sum to 1 and Assumptions~\ref{assump_r_t_bounded} and~\ref{assump_d_t_bounded} to:
\begin{multline*}
\left| \frac{\inner{\hat{R}'_H}{\piaux}}{\inner{\hat{D}'_H}{\piaux}} - \frac{\inner{\hat{R}}{\piorg}}{\inner{\hat{D}}{\piorg}} \right| \leq\\
\leq \varepsilon^{-2} \cdot \left| \inner{\hat{R}'_H}{\piaux} - \inner{\hat{R}}{\piorg} \right| \cdot \dmax + \varepsilon^{-2} \cdot \rmax \cdot \left| \inner{\hat{D}}{\piorg} - \inner{\hat{D}'_H}{\piaux} \right| \,\, . 
\end{multline*}
Then, using Lemma~\ref{lemma_expected_reward_diff_cont_equal}, 
\begin{multline*}
\left| \frac{\inner{\hat{R}'_H}{\piaux}}{\inner{\hat{D}'_H}{\piaux}} - \frac{\inner{\hat{R}}{\piorg}}{\inner{\hat{D}}{\piorg}} \right| \leq\\
\leq\varepsilon^{-2} \cdot \bigO{\frac{1}{H}} \cdot \dmax + \varepsilon^{-2} \cdot \rmax \cdot \bigO{\frac{1}{H}} \,\, . 
\end{multline*}
Overall, by using the triangle inequality for the bound we obtained above and Lemmas~\ref{simplified_revorg} and~\ref{simplified_revaux} we deduce:
\begin{align*}
|\revaux(H) - \revorg| = \bigO{\frac{1}{H}} \, \, . &\qedhere
\end{align*}
\end{proof}


    \section{Performance} \label{section:performance} 

Having proven the theoretical approximation error, we proceed to evaluate the practical performance of PTO by comparing it against OSM. 
In order to do so we evaluate the running time of both methods for Bitcoin by using existing code of OSM for Bitcoin~\cite{osmgithub}.
We begin with an overview of the MDP model for Bitcoin~(\S\ref{section:performance:model}).
We validate the results of our implementation~(\S\ref{section:performance:validation}), slightly improving the known threshold for Bitcoin. 
We investigate the effects of different hyperparameters on the optimal revenue obtained by PTO~(\S\ref{section:performance:hyperparameters}).
Finally, we compare the running times of PTO and OSM and show PTO to be about~10 times faster~(\S\ref{section:performance:runningTime}). 


        \subsection{Bitcoin Model}\label{section:performance:model}

The ARR-MDP for bitcoin is similar to the model of Sapirshtein et al.~\cite{sapirshtein2016optimal}. We first describe 3 important parameters for the model, then describe the action space, the state space and the transitions. 

As in a general blockchain protocol, there is a rational miner who wants to maximize her revenue.
The rest of the network is represented by honest miners who follow the prescribed protocol. 

We assume the objective function is the ratio between the number of blocks of the rational miner and the total number of blocks of the entire network.
Because Bitcoin is a PoW blockchain, the \emph{relative mining power} determines the probability of the miner to mine a new block.
We denote this parameter $\alpha \in [0, \frac{1}{2})$.

The Bitcoin protocol specifies that in the case of a tie in the longest chain rule, the tie is decided in favor of the first chain the miner saw.
In case of a tie between the rational miner's chain and some other chain, we assume the rational miner's block is received first by a fraction $\gamma \in [0, 1]$ of the network.
This is called the \emph{rushing level} of the miner and determines what fraction of the honest miners will keep mining on top of the rational miner's chain.
This determines the probability that the next block will be mined on top of the miner's chain and then her chain will be chosen by all.

We assume that the miner mines on a single secret chain and that the miner will not choose to challenge blocks before the \emph{last fork}~-- the blocks following the last block common to both the miner's chain and the current public chain, as in~\cite{sapirshtein2016optimal}.
We also assume forks between honest miners never occur as in~\cite{sapirshtein2016optimal, eyalmajority, nayak2016stubborn}.

Denote by~$a$ the length of the miner's secret chain.
Denote by~$h$ the number of blocks in the public chain since the last fork.
In order to obtain a finite state space MDP, we cannot consider all possible strategies of the rational miner.
We assume there is a maximum possible length for both~$a$ and~$h$ as in~\cite{sapirshtein2016optimal}.
We call this bound the \emph{maximum fork length}.

        
        \subsubsection{Action Space}

We now describe the possible actions the miner can choose.
\begin{enumerate}
    \item \textsf{Adopt}~-- The miner chooses to abandon her private chain and accept the current public chain.
    
    \item \textsf{Override}~-- The miner reveals the first~$h + 1$ blocks from her private chain and overtakes the public chain.
    This action is possible only when~$a > h$.
    
    \item \textsf{Match}~-- The miner reveals the first~$h$ blocks of her private chain and matches the public chain.
    This action is possible only when $a = h$ and when the last block mined was by someone other than the miner.
    This symbolizes the case where the miner hears about a newly mined block and then quickly reveals a block of the same height mined in advance.
    This triggers a split in the network determined by the miner's rushing level.
    Each honest miner chooses to mine on the first chain she sees.
    
    \item \textsf{Wait}~-- The miner does not reveal blocks and keeps mining on her private chain.
\end{enumerate}


        \subsubsection{State Space}

The \textsf{match} action gives 3 cases for the current state, which need to be differentiated.

\begin{enumerate}
    \item \textsf{Irrelevant}~-- The last block was mined by the rational miner. \textsf{Match} cannot be performed because the miner just mined a block so she does not have a block prepared in advance.
    \item \textsf{Relevant}~-- The last block was mined by an honest miner. \textsf{Match} can be performed if $a \geq h$.
    \item \textsf{Active}~-- \textsf{Match} was already performed and the network is split so \textsf{match} cannot be performed again. 
\end{enumerate}
Denote by $\textit{fork}$ the state of the system from this list. 
Then, the states in the MDP are represented by a vector with 3 elements: $(a, h, \textit{fork})$.


        \subsubsection{Transitions}

If the miner chooses to \textsf{wait} when \textit{fork} is not \textsf{active}, either~$a$ increases by~1 w.p~$\alpha$ or~$h$ increases by~1 w.p~$1 - \alpha$ and \textit{fork} is updated to \textsf{irrelevant} or \textsf{relevant} respectively.

If the miner chooses to \textsf{adopt}, both $a$ and~$h$ become~0 and $D_t = h$ as the miner chooses to accept~$h$ blocks.
Otherwise, if the miner chooses to \textsf{override}, $a \gets a-h-1$, $h \gets 0$ and $R_t = D_t = h+1$ as the miner appends~$h+1$ blocks to the blockchain.

If the miner chooses to \textsf{match}, $\textit{fork} \gets\textsf{active}$ as the miner causes a fork in the network.
If the miner chooses to \textsf{wait} when $\textit{fork}$ is \textsf{active} either:
\begin{enumerate}
    \item $a$ increases by~1 w.p~$\alpha$ as the rational miner mines a new secret block,
    \item $a \gets a-h$, $h \gets 1$ and~$\textit{fork} \gets \textsf{relevant}$ w.p~$\alpha \gamma$ as an honest miner mines on top of the rational miner's chain, thus solving the fork and providing $R_t=D_t=h$, or,
    \item $h \gets h+1$ and~$\textit{fork} \gets \textsf{relevant}$ w.p~$\alpha (1-\gamma)$ as an honest miner mines on top of the public chain.
\end{enumerate}

In order to enforce the maximum fork length, for certain states we forbid actions which may lead to~$a$ or ~$h$ increasing too much.
Note that, blocks mined are counted towards the reward and the difficulty contribution if and when both the rational miner and the rest of the network accept them.


        \subsection{Optimal Strategies} \label{section:performance:validation} 

To validate PTO, we compare the revenue of the optimal strategy found by PTO against the state of the art OSM~\cite{sapirshtein2016optimal}. 
We compare for different $\alpha$'s and $\gamma=0$, with a maximum fork length of~95, similar to OSM. 
We choose expected horizon to be~$H=10^6$ and we use policy iteration (for SSP) with a stopping threshold of~$10^{-5}$. 
The stopping threshold is a parameter for policy iteration which specifies the desired approximation error for the algorithm~\cite{bertsekas1995dynamic}.

\begin{table}
\centering
\begin{tabular}{ S[table-format=1.3] S[table-format=1.5] S[table-format=1.5] } 
\toprule
{Power ($\alpha$)} & {PTO Revenue ($\revorg$)} & {OSM Revenue} \\
\midrule
{1/3} & 0.33705 & 0.33705\\
\hline
0.35 & 0.37077 & 0.37077\\
\hline
0.375 & 0.42600 & 0.42600\\
\hline
0.4 & 0.48866 & 0.48866\\
\hline
0.425 & 0.56809 & 0.56808\\
\hline
0.45 & 0.66894 & 0.66891\\
\hline
0.475 & 0.80184 & 0.80172\\
\bottomrule
\end{tabular}
\caption{Comparison between PTO with policy iteration and OSM~\cite{sapirshtein2016optimal} for Bitcoin when limiting the maximum fork length to 95 and using~$\gamma = 0$.}
\label{table:bitcoin_results}
\end{table}

Table \ref{table:bitcoin_results} summarizes the results. 
As expected, PTO reproduces the results of OSM. 
Furthermore, some results for PTO are higher than the results for~OSM by more than the~$10^{-5}$ approximation error. 
This is due to a bug in the original OSM code. 
For these values, OSM stopped by mistake after hitting a hard-coded maximal number of iterations. 
We removed the redundant stopping condition before comparing the running times. 


        \subsection{Hyperparameters} \label{section:performance:hyperparameters} 

We empirically analyze the effect PTO's hyperparameters. 


            \subsubsection{Expected Horizon}

We first consider the expected horizon. 
In order to do so, we experiment with Bitcoin for several values of~$\alpha$, with~$\gamma=0.5$, with a fixed maximum fork length of 50 and using policy iteration with a stopping threshold of~$10^{-5}$. 

Figure \ref{fig:expected_horizon} plots the revenue of the policy found against the horizon length~$H$.
The figure shows that as the expected horizon increases, the revenue converges to the revenue of the optimal policy.
Since different~$\alpha$'s were considered, all the revenues for each $\alpha$ were normalized by the best revenue achieved by PTO for said~$\alpha$.

The graph shows that as the expected horizon is increased, the approximate revenue converges quickly to its optimal value,
as supported by Theorem \ref{main_theorem}. 

Intuitively, one might expect that for higher values of $\alpha$, our method would require higher values of $H$ to converge,
since more powerful miners are more capable of creating longer forks, and thus should require more consideration towards future blocks by looking into a further expected horizon.
\begin{figure}
\includegraphics[width=\linewidth]{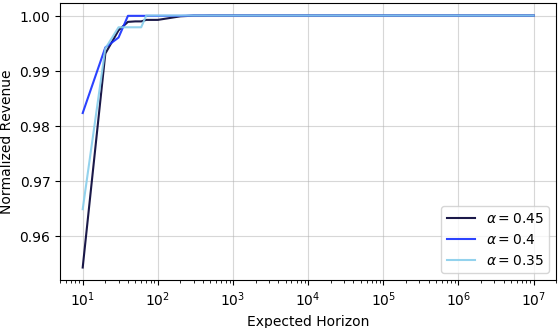}
\caption{The effect of the expected horizon~$H$. We plot results for Bitcoin with different~$\alpha$'s, $\gamma = 0.5$, and maximum fork of 50, normalized by the maximum revenue achieved for each~$\alpha$.}
\label{fig:expected_horizon}
\end{figure}

Our results in Figure \ref{fig:expected_horizon} indicate that this is not necessarily the case. 
We attribute this find to the fact that different values of~$\alpha$ have different optimal policies, and thus are affected differently by the expected horizon.


\paragraph{Note}
A drawback of PTO when compared to OSM is that although the approximation error in PTO is tight and decreases linearly with the expected horizon, it is not known in advance.
OSM, on the other hand, takes the approximation error as a parameter and outputs an~$\varepsilon$-optimal policy.
However, as we see in Figure \ref{fig:expected_horizon}, choosing a reasonable expected horizon such as~$10^6$ ensures a negligible approximation error.


            \subsubsection{Maximum Fork Length}\label{subsection:performance:maximum_fork_length_analysis}
        
\begin{figure}
\includegraphics[width=\linewidth]{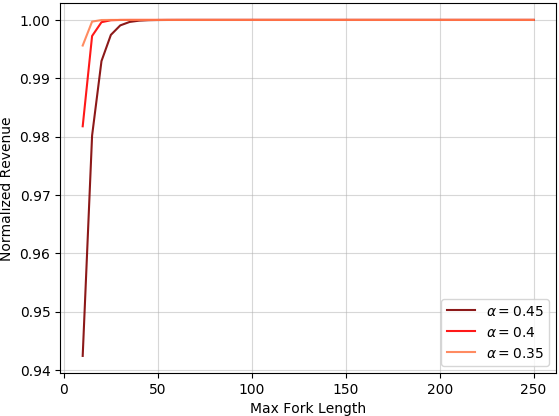}
\caption{The effect of the maximum fork length. We plot results for Bitcoin with different $\alpha$'s, $\gamma = 0.5$ and $H=10^4$, normalized by the maximum revenue achieved for each $\alpha$.}
\label{fig:max_fork}
\end{figure}
        
We next investigate the effect of the maximum fork length.
We experiment with Bitcoin for several values of~$\alpha$, with a fixed expected horizon of~$10^4$, and using policy iteration with a stopping threshold of~$10^{-5}$.
The revenues obtained are presented in Figure \ref{fig:max_fork}, normalized for different $\alpha$ as above. 
As the maximum fork is increased, the agent has more possible actions in the game, and thus the revenue is non decreasing, as clearly evident in the figure.
In addition, the higher $\alpha$ is, the higher the maximum fork length needed to achieve the same approximation, as evident in the figure by the dependence of the fork length required to reach optimal performance on $\alpha$.

%
The maximum fork length enables the model to disregard possible policies which allow forks longer than some threshold.
This ensures the state space is finite and enables us to use PTO.
Ignoring feasible policies seemingly hurts the optimality of our results.
However, we ignore longer forks because the probability of the miner to obtain such a fork declines exponentially with the fork length.
Therefore, considering a large enough maximum fork length such as~100 is more than enough to obtain a good approximation; increasing it further yields negligible improvement, as shown empirically in Figure \ref{fig:max_fork}. 


        \subsection{Running Time} \label{section:performance:runningTime} 

We show that PTO with policy iteration is faster than OSM.
We compare the methods for different maximum allowed fork lengths in order to compare the methods for different sizes of the state space. 
This allows us to conjecture about the efficiency of generalizing the methods to other blockchains with larger state spaces (e.g., Ethereum).
In order to perform the comparison, we use the OSM code generously shared by Ren Zhang~\cite{osmgithub}.

Running time heavily depends on the low level implementation and platform. 
Instead, we compare the number of linear system that each method solves. However, this comparison too depends on the dynamic programming algorithm used to solve the (standard) MDPs. We found that in contrast to the relative value iteration used in the original OSM work, the policy iteration algorithm works significantly better both for OSM and PTO. 
\begin{figure}
\includegraphics[width=\linewidth]{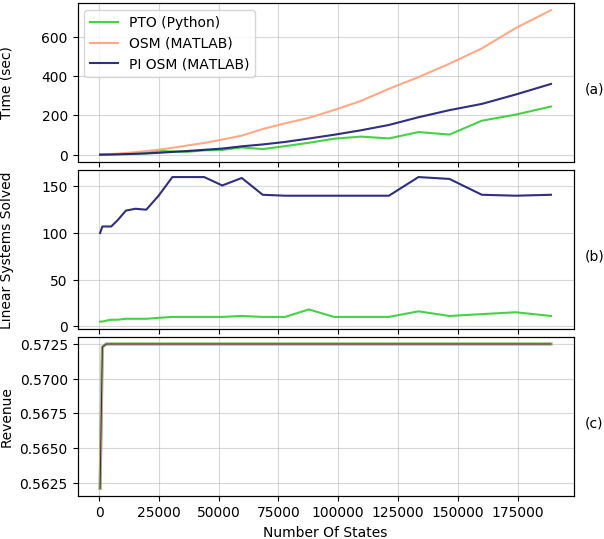}
\caption{(a): Comparison of the running time of the 3 methods for different maximum fork lengths.
The maximum fork length determines the number of states in the model.
(b): Comparison of the number of times a linear system was solved.
This enables comparison of the 2 methods even though they run on different environments.
(c): Comparison of the revenues obtained by the 3 methods.
PTO achieves similar results  to both other methods.}
\label{fig:max_fork_timings}
\end{figure}

We therefore improved OSM to use policy iteration (for the average reward criterion) instead of relative value iteration, and we call this implementation \emph{Policy-Iteration OSM} (\emph{PI-OSM}). 
Figure~\ref{fig:max_fork_timings}(a) shows PI-OSM has a shorter running time than~OSM when run in the same environment. 
Next, we compare the number of linear systems solved for PTO and for PI-OSM.
Figure \ref{fig:max_fork_timings}(b) shows that the number of linear systems solved for PTO is about one order of magnitude smaller than PI-OSM.

When solving MDPs there is a trade off between accuracy and running time.
For a fair comparison, we made sure that PTO was at least as accurate as OSM and PI-OSM.
Figure~\ref{fig:max_fork_timings}(c) shows that the results obtained by all methods were of similar accuracy (up to about $10^{-6}$).


    \section{Ethereum} \label{section:ethereum}

After evaluating the performance of PTO we now move on to find optimal strategies in Ethereum. 
We first describe our MDP model for Bitcoin~(\S\ref{section:ethereum:model}).
We then present the results of PTO for Ethereum and compare them to the approximately optimal results of SquirRL~(\S\ref{section:ethereum:results}).
Finally, we describe our derivation of the security threshold for Ethereum~(\S\ref{section:ethereum:security_threshold}). 

        \subsection{Model}\label{section:ethereum:model}

We now describe an ARR-MDP model for Ethereum~\cite{buterin2013ethereum, wood2014ethereum}.

A full implementation is available in the GitHub repository.

We describe the relevant differences from Bitcoin. 
Then, we describe the action space, the state space and the transitions of the ARR-MDP. 

The basic mechanics of the protocol are similar~-- in both systems miners generate blocks that form a graph. 
%
As in Bitcoin, the relative mining power~$\alpha$ and the maximum fork length will play an important part. 

As in Bitcoin, Ethereum's blockchain is the longest chain of blocks. 
However, ties are broken uniformly at random. 
So, the miner's rushing level does not play a role in Ethereum and we can assume a constant $\gamma=0.5$. 

Ethereum presents the concept of \emph{uncle blocks}~\cite{buterin2013ethereum}.
In order to compensate a miner of a block that ended up out of the main chain, Ethereum introduces additional rewards, as follows. 
A block in the main the chain may reference a previous block, called an uncle block, if it is a direct child of a previous block in the main chain.
In this context, the referencing block is called a \emph{nephew block}.
A block can reference up to~2 such uncle blocks.

The miner of the uncle block receives an \emph{uncle reward}~\cite{ritz2018impact}.
This reward depends on the \emph{uncle distance}, that is, the number blocks in the chain from the last fork to the nephew block.
The uncle reward starts at~$\frac{7}{8}$ of the regular block reward if the nephew is the first block since the last fork between the nephew's chain and the uncle block, we refer to this as an uncle with a distance of~1.
The reward decreases by~$\frac{1}{8}$ for any additional block between the last fork and the nephew down to~$\frac{2}{8}$ of the regular block reward.
This happens when the nephew is the~6\textsuperscript{th} block since the last fork as an uncle distance of more than 6 is not allowed. 

The miner of the nephew block receives a \emph{nephew reward} in addition to the regular block reward.
This is in contrast to receiving only the block reward in Bitcoin.
The nephew reward is equal to~$\frac{1}{32}$ of the regular block reward.

Originally, Ethereum did not count uncles for its difficulty adjustment mechanism.
However, it allowed strategies which directly exploit this to intentionally create blocks destined to become uncles~\cite{ritz2018impact}.

In order to deter miners from doing so, Ethereum was updated to take into account uncle blocks for the difficulty~\cite{ethereumdifficultychange}.
However for its code's backwards compatibility, the actual implementation specifies that in the case the block references either one or two uncles, the block and its uncles count only as~2 blocks towards the difficulty adjustment.
This allows situations in which~3 blocks in total count into the difficulty adjustment as~2 blocks.
This is instead of the change which was originally intended.

In the model, we follow the intended mechanism and assume that any uncle counts as~1 block regardless of whether~2 uncles were referenced by a single block.
This is in line with other previous works~\cite{ritz2018impact,grunspan2019selfish}. 


The Ethereum protocol specifies that honest miners reference all available uncle blocks which were not referenced before, and do so in order~-- reference further uncles first \cite{wood2014ethereum}.
Therefore, the honest miners reference all possible uncles, including blocks by the rational miner.

The rational miner however, can choose which blocks to reference as uncles. 
Since when referencing an uncle the nephew reward is dwarfed by the resulting contribution to the difficulty, we restrict the rational miner to only reference her own blocks as uncles when possible as in~\cite{hou2019squirrl}.

        
            \subsubsection{Action Space}

The Ethereum model uses the 4 actions defined in the Bitcoin model: \textsf{adopt}, \textsf{override}, \textsf{match}, \textsf{wait}.

We also introduce a new action available to the rational miner.
If~$a>0$ and~$h>1$, the miner may \textsf{reveal} the first block of her private chain to be included as an uncle unless it was already revealed previously.

Revealing additional blocks will not achieve anything as an uncle has to be direct child of a block in the main chain.
Therefore, we assume the miner never reveals more than one block except for when she chooses to \textsf{override} or \textsf{match}.


        \subsubsection{State Space}

Same as for the Bitcoin model, denote by~$a$ the length of the miner's secret chain and by~$h$ the number of blocks in the public chain since the last fork.

Because of the change to the prescribed policy in a case of a tie in the longest chain rule, the match action now gives only~2 cases for the current state \cite{sapirshtein2016optimal, hou2019squirrl}.
\begin{enumerate}
    \item \textsf{Relevant}~-- \textsf{Match} can be performed if $a \geq h$.
    \item \textsf{Active}~-- \textsf{Match} was already performed and the network is split so \textsf{match} cannot be performed again. 
\end{enumerate}
Denote $\textit{fork}$ as the current case of the system as described above.

We now move on to describe how to capture the state regarding the uncle blocks.
We distinguish between the rational miner's blocks and blocks mined by other miners.

Let~$u_h$ be a binary vector of length 6 denoting whether there are blocks by honest miners since the last fork that can be included as uncles. 
$u_h$ registers only the last 6 possible uncles since further uncles are not allowed to be included.
Each entry~$i$ in this vector denotes whether there is possible block to be included as an uncle with distance~$i$ in the first block after the last fork between the rational's miner chain and the public chain.
Note that this is a different fork than the one between the possible uncle and its nephew.

We use two variables to denote 
the rational miner's uncle state. 
Let~$r$ denote the length of the public chain since the last fork when the rational miner revealed the first block of her secret chain minus~1~($h-1$) or~0 if it is still a secret.
If the miner's first block will be included as an uncle, $r$ specifies its uncles distance. This is because the miner's first block since the last fork cannot be referenced before as it was still a secret when the first $r+1$ public blocks since the last fork have been published.

In addition, let~$u_a$ denote whether there is a revealed block by the rational miner from before the last fork which was not referenced as an uncle before.
$u_a$ will serve only to denote whether there is a potential uncle block pending from before the last fork.
Unlike $r$ or $u_h$, it does not capture its uncle distance since we count its reward and difficulty contribution immediately after the its fork is resolved.
We can do this since we know for sure that this block will be referenced as the rational miner references her own blocks and the rest of the network reference all potential uncles.

The states in the MDP are represented by a vector with 6 elements: $(a, h, \textit{fork}, r, u_a, u_h)$. 

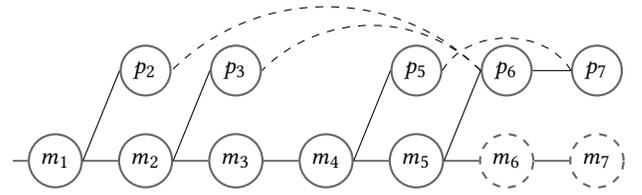
\begin{figure}
    \begin{tikzpicture}[node distance=12mm, circlenode/.style={circle,
                        draw=black!60, thick, minimum size=5mm},
                        dashedcirclenode/.style={circle,
                        draw=black!60, thick, dashed, minimum size=5mm}]
        
        \node[circlenode] (m1) {$m_1$};
        \node[circlenode] (m2) [right of=m1] {$m_2$};
        \node[circlenode] (m3) [right of=m2] {$m_3$};
        \node[circlenode] (m4) [right of=m3] {$m_4$};
        \node[circlenode] (m5) [right of=m4] {$m_5$};
        \node[dashedcirclenode] [right of=m5] (m6) {$m_6$};
        \node[dashedcirclenode] [right of=m6] (m7) {$m_7$};
        
        \node[circlenode] (p2) [above of=m2] {$p_2$};
        \node[circlenode] (p3) [above of=m3] {$p_3$};
        \node[circlenode] (p5) [above of=m5] {$p_5$};
        \node[circlenode] (p6) [above of=m6] {$p_6$};
        \node[circlenode] (p7) [above of=m7] {$p_7$};
        
        \node[left=2mm of m1.west] (start) {};
        
        \draw[-] (m1.west) to (start);
        \draw[-] (m2.west) to (m1.east);
        \draw[-] (m3.west) to (m2.east);
        \draw[-] (m4.west) to (m3.east);
        \draw[-] (m5.west) to (m4.east);
        \draw[-] (m6.west) to (m5.east);
        \draw[-] (m7.west) to (m6.east);
        
        \draw[-] (p2.west) to (m1.east);
        \draw[-] (p3.west) to (m2.east);
        \draw[-] (p5.west) to (m4.east);
        \draw[-] (p6.west) to (m5.east);
        \draw[-] (p7.west) to (p6.east);

        \draw[dashed] (p6.west) to [bend right=45] (p2.east);
        \draw[dashed] (p6.west) to [bend right=45] (p3.east);
        \draw[dashed] (p7.west) to [bend right=60] (p5.east);
    \end{tikzpicture}
    \caption{An example state for Ethereum.
    Blocks mined by honest miners are marked by~$p_i$.
    Dashed lines present uncle references and only appear for $p_6$ and $p_7$.
    Blocks mined by the rational miner are marked by~$m_i$.
    Secret blocks are marked by a dashed outline.
    }
    \label{fig:eth_example_1}
\end{figure}

Figures~\ref{fig:eth_example_1} and~\ref{fig:eth_example_2} give~2 examples of possible states.
In Figure~\ref{fig:eth_example_1},~$p_6$ and~$m_6$ are the last fork.
As there are~2 blocks in each chain since the last fork,~$a=2$ and~$h=2$.
There are 3 uncle references marked: the block~$p_6$ references uncles~$p_2$ and~$p_3$ with uncles distances of~4 and~3 respectively and the block~$p_7$ references the uncle~$p_5$ with a distance of~2.
The honest uncles vector~$u_h$ is equal to~$(1,0,1,1,0,0)$ because it captures the distance of previous uncles with respect to the last fork.
As the miner's blocks are secret, $r=0$.
If the miner chooses to \textsf{reveal} then $r$ would become~1.
Because there are no potential uncle blocks of the rational miner before the last fork,~$u_a=0$.

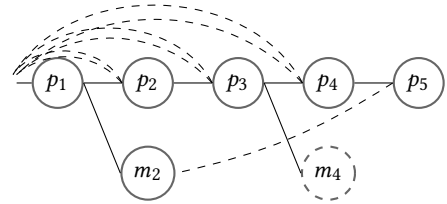
\begin{figure}
    \begin{tikzpicture}[node distance=12mm, circlenode/.style={circle,
                        draw=black!60, thick, minimum size=5mm},
                        dashedcirclenode/.style={circle,
                        draw=black!60, thick, dashed, minimum size=5mm}]
        
        \node[circlenode] (p1) {$p_1$};
        \node[circlenode] (p2) [right of=p1] {$p_2$};
        \node[circlenode] (p3) [right of=p2] {$p_3$};
        \node[circlenode] (p4) [right of=p3] {$p_4$};
        \node[circlenode] (p5) [right of=p4] {$p_5$};
        
        \node[circlenode] (m2) [below of=p2] {$m_2$};
        \node[dashedcirclenode] [below of=p4] (m4) {$m_4$};
        
        \node[left=2mm of p1.west] (start) {};
        
        \draw[-] (p1.west) to (start);
        \draw[-] (p2.west) to (p1.east);
        \draw[-] (p3.west) to (p2.east);
        \draw[-] (p4.west) to (p3.east);
        \draw[-] (p5.west) to (p4.east);
        
        \draw[-] (m2.west) to (p1.east);
        \draw[-] (m4.west) to (p3.east);
        
        \draw[dashed] (p2.west) to [bend right=45] (start);
        \draw[dashed] (p2.west) to [bend right=60] (start);
        \draw[dashed] (p3.west) to [bend right=45] (start);
        \draw[dashed] (p3.west) to [bend right=60] (start);
        \draw[dashed] (p4.west) to [bend right=45] (start);
        \draw[dashed] (p4.west) to [bend right=60] (start);
        \draw[dashed] (p5.west) to [bend left=10] (m2.east);
    \end{tikzpicture}
    \caption{An example state for Ethereum.
    This figure is with the same conventions as in Figure~\ref{fig:eth_example_1}.
    Each of the blocks~$p_2$,~$p_3$ and~$p_4$ references 2 uncles previous to block $p_1$.
    }
    \label{fig:eth_example_2}
\end{figure}

In Figure~\ref{fig:eth_example_2},~$p_4$ and~$m_4$ are the last fork, $a=1$ and~$h=2$.
The honest uncles vector~$u_h$ does not include the uncles already referenced by~$p_2$ and~$p_3$ as these uncles were already referenced before the last fork.
It includes only the uncles referenced by~$p_4$. However as the figure does not show the distances of those uncles $u_h$ is ambiguous.
As the miner's block is secret, $r=0$.
As $m_2$ is revealed, it is a pending uncle block of the rational miner which was mined before the last fork, thus~$u_a=1$.

Figure~\ref{fig:eth_example_2} also demonstrates why there can be at most~1 pending uncle block of the rational miner before the last fork.
The only case in which there could be a pending uncle block is when the miner has at least~1 revealed block, chooses to \textsf{adopt} and if all the blocks in the public chain reference previous uncle blocks and there is no room left for referencing the rational miner's block.
For there to be~2 pending uncle blocks, this would have to happen twice.
However, if~$a\leq1$, the rational miner would never choose to \textsf{adopt} unless~$h\geq2$. As if $h<2$ \textsf{match} or \textsf{override} are strictly more profitable.
So, there would have to be at least 4 blocks which all reference previous uncles in $u_h$.
But, this is not possible.
Since we assume forks between the honest miners never occur, and as each block can reference~2 uncles,~3 honest blocks are enough to reference all the 6 previous uncles of~$u_h$.
This means the rational miner's block can be referenced in the 4\textsuperscript{th} block ($p_5$).


        \subsubsection{Transitions}

The transitions in Bitcoin leading to changes in ~$a$, $h$ and~$\textit{fork}$ are similar in Ethereum.
We now describe how the transitions change~$r$,~$u_a$ and~$u_h$.

If the miner chooses to \textsf{wait} when $\textit{fork}$ is not \textsf{active}, the information stays the same.
If the miner chooses to \textsf{reveal} then $r \gets h-1$.
If the miner chooses to \textsf{match} when $\textit{fork}$ is \textsf{relevant} then also $r \gets h-1$.

If the miner chooses to \textsf{adopt}, then the public chain of length~$h$ is accepted and all possible uncles in~$u_a$ and~$u_h$ might be referenced.
If $r>0$ then the first block of the rational miner's chain can also be referenced.

The actual number of uncles referenced depends on~$h$ as each block has a maximum of~2 possible uncles to reference.
All referenced uncles are removed from $u_h$~and the remaining possible uncles in~$u_h$ are shifted back by~$h$.
Also, $u_a \gets 0$.
If $r>0$ and the miner's first block was not referenced then we mark $u_a \gets 1$ to remember this uncle still has to be referenced and then set $r \gets 0$.
The difficulty contribution is then the number of blocks~$h$ plus the number of uncle blocks referenced.

If~$r$ was more than~0, we reward the miner with a relevant uncle reward assuming it will be referenced in the first block of the next fork.
This happens even if the miner's uncle block was currently not included.
This block also counts for the difficulty contribution regardless. 

If~$u_a$ was~1 and at least~2 honest uncle blocks were included, this means that we counted the uncle reward of the pending uncle based on a shorter uncle distance.
We correct this by fining the miner by $\frac{1}{8}$.
An example of this case is if the miner chooses to \textsf{adopt} after the state illustrated in Figure~\ref{fig:eth_example_2}.

If the miner chooses to \textsf{override}, all the uncles in~$u_h$ are shifted back by~$h+1$.
Also,~$r \gets 0$ and~$u_a \gets 0$, as the miner appends~$h+1$ blocks to the blockchain.
If~$u_a$ was~1 before, we give the miner a nephew reward as well and add 1 to the difficulty contribution in addition to $h+1$ as described in the Bitcoin model.

If the miner chooses to \textsf{wait} when $\textit{fork}$ is active, and if the next mined block is by the rational miner or by an honest miner who extends the public chain no blocks become accepted by everyone so the fork is not resolved. In this case, the uncle information stays the same.
However, if the next mined block is by an honest miner who extends the miner's chain.
The fork is resolved and the public advances by~$h$ blocks.
In this case, all the uncles in~$u_h$ are shifted back by~$h$,~$r \gets 0$ and~$u_a \gets 0$.
Also in this case, if~$u_a$ was~1 before, we give the miner a nephew reward as well and add 1 to the difficulty contribution in addition to $h$ as described in the Bitcoin model.

As in Bitcoin, blocks, uncles and nephew rewards are counted towards the reward and the difficulty contribution only if and when they are agreed upon by both the rational miner and the rest of the network.
The one exception to this rule is when we count the reward and difficulty contribution in advance in the case of $u_a=1$.
This complicates the model but reduces the number of states significantly.

The model in~\cite{hou2019squirrl,squirrlgithub} captures the uncle information using both a parameter similar to $r$ and a ternary vector similar to $u_h$ in which every element registers whether there is a possible uncle of the rational miner, a possible uncle of an honest miner or no potential uncle.

Counting the rational miner's uncles in advance as in our model reduces the state space size by a factor of $\frac{3^6}{2 \cdot 2^6} \approx 5.7$.


\subsection{Results}\label{section:ethereum:results}

\begin{table}
\centering
\begin{tabular}{ S[table-format=1.3] S[table-format=1.6] } 
\toprule
{Power ($\alpha$)} & {Revenue ($\revorg$)}\\
\midrule
0.25 & 0.250705\\
\hline
0.275 & 0.282596\\
\hline
0.3 & 0.317798\\
\hline
0.325 & 0.359305\\
\hline
0.35 & 0.407925\\
\hline
0.375 & 0.465532\\
\hline
0.4 & 0.534359\\
\hline
0.425 & 0.618737\\
\hline
0.45 & 0.718527\\
\hline
0.475 & 0.826861\\
\bottomrule
\end{tabular}
\caption{PTO results for Ethereum when $H=10^5$ and limiting the maximum fork length to 20.}
\label{table:ethereum_results}
\end{table}

Table \ref{table:ethereum_results} shows results of PTO for Ethereum for various $\alpha$'s and maximum fork length of 20.
We chose the expected horizon to be~$H=10^5$ and used policy iteration with a stopping threshold of~$10^{-7}$.

\begin{figure}
    \includegraphics[width=\linewidth]{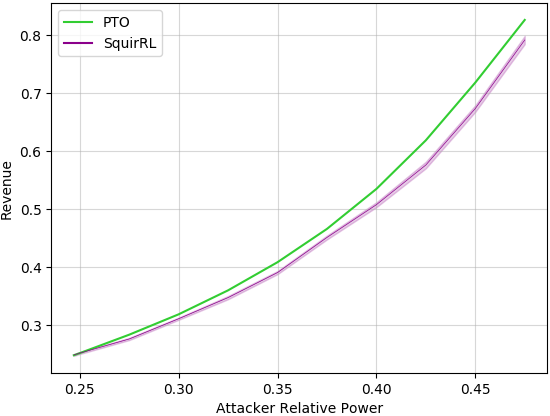}
    \caption{Comparison between the results for Ethereum of PTO and SquirRL with different $\alpha$'s and maximum fork length of 20.
    The results of SquirRL are surrounded by a confidence interval with a confidence of 0.99.}
    \label{fig:squirrl_comparison}
\end{figure}

Figure \ref{fig:squirrl_comparison} gives a comparison between the results in Table~\ref{table:ethereum_results} and SquirRL.
To do so, we used the code of SquirRL generously shared with us by the authors \cite{squirrlgithub}. 
SquirRL relies on Q-Learning with value function approximation (with neural nets) and they used Monte Carlo approximation to calculate the revenue of the policy found.
This only gives a confidence interval for the revenue of the policy and not an exact number.
This is still a good approximation as seen in the tight confidence band in the figure.

As seen in the figure, PTO outperforms SquirRL for all cases.
Thanks to Theorem \ref{main_theorem} we know that PTO converges to the optimum.

But, the difference between the results is also affected by the slight differences in the model.
Our model allows the rational miner to reveal her first block of the private chain at any time she desires in order for it to be counted as an uncle.
In addition, in our model the objective of the miner is the miner's reward divided by the total number of blocks counted towards the difficulty (blocks in the main chain and uncles) rather than the miner's relative revenue as was used by~\cite{squirrlgithub}.
SquirRL's results in the figure represent the revenue of the miner in our model using the policy obtained by SquirRL in order to get a relevant comparison.

\paragraph{Note}
Since Q-Learning with value function approximation is not guaranteed to converge to the optimum, different runs of SquirRL may give different policies.
Figure \ref{fig:squirrl_comparison} only shows the confidence interval of one policy found and might not be reproducible.

\paragraph{Note}
The state space in Ethereum is much larger than the state space in Bitcoin and since the transition matrix is of size~$|\fancyS|~\times~|\fancyS|$, it is too big to fit in the memory of a reasonably powerful server.
In order to overcome this, SquirRL used Q-Learning with value function approximation which does not require the transition matrix explicitly.
We used a different approach by using a sparse matrix.
Although the transition matrix is of a large size, the number of possible transitions from each state is relatively small and does not depend on the state size.
This means that most of the entries in the transition matrix are zeros and using a sparse matrix uses memory linear in the size of the state space instead.


        \subsection{Security Threshold}\label{section:ethereum:security_threshold}
         
\begin{table}
\centering
\begin{tabular}{ S[table-format=1.4] S[table-format=1.6] S[table-format=1.6] } 
\toprule
{Power ($\alpha$)} & {Revenue ($\revorg$), 20} & {Revenue ($\revorg$), 25}\\
\midrule
0.24 & 0.24 & 0.24\\
\hline
0.245 & 0.245 & 0.245\\
\hline
0.246 & 0.246 & 0.246\\
\hline
0.2467 & 0.2467 & 0.2467\\
\hline
0.2468 & 0.246806 & 0.246807\\
\hline
0.2469 & 0.246928 & 0.246928\\
\hline
0.247 & 0.247050 & 0.247050\\
\bottomrule
\end{tabular}
\caption{PTO results for Ethereum when $H=10^6$ and limiting the maximum fork length to 20 or 25.}
\label{table:ethereum_results_2}
\end{table}
         
Table \ref{table:ethereum_results_2} shows higher accuracy results for Ethereum for $\alpha$'s around the security threshold with a maximum fork length of either~20 or~25.
In order to obtain more accurate results, we chose expected horizon to be~$H=10^6$ and a stopping threshold of~$10^{-8}$.

The table shows that the new threshold found with PTO is~0.2468.
This is lower than previous upper bounds found which are approximately~0.26~\cite{ritz2018impact,feng2019selfish}.

The empirical analysis in Section~\ref{subsection:performance:maximum_fork_length_analysis} suggests that for smaller values of $\alpha$ only small maximum fork lengths are necessary to obtain a good approximation.
This can also be seen in Table \ref{table:ethereum_results_2} as increasing the fork length barely changed the results.
Therefore, we consider the threshold we found to be a good approximation.


    \section{Conclusion} 

We presented PTO: an efficient method to find optimal mining strategies in PoW blockchains. 
PTO forms a probabilistically terminating state machine that can be optimized directly to find the desired strategy. 
We prove PTO is correct and bound its approximation error by $O(\frac{1}{H})$.
PTO is an order of magnitude more efficient than the state of the art. 
We use it to calculate the security threshold of Ethereum and reduce it from~0.26~\cite{grunspan2019selfish,feng2019selfish} to~0.2468. 

PTO applies to any blockchain protocol that can be modeled as an~MDP with an ARR objective function. 
Due to its speed, it can be used repeatedly to find more robust reward schemes with higher security thresholds. 

    
    \begin{acks}
    
We thank Ren Zhang for sharing his code of OSM \cite{osmgithub} and the authors of SquirRL~\cite{squirrlgithub} for sharing their code.

This research was partially supported by the Israel Science Foundation (grants No.\ 1641/18 and No.\ 759/19) and the Open Philanthropy Project Fund, an advised fund of Silicon Valley Community Foundation.

\end{acks}

\bibliographystyle{ACM-Reference-Format}
\bibliography{refs}



\appendix
        \section{Proofs}\label{section:proofs_appendix}
In this appendix we restate and prove all the deferred proofs from section~\ref{section:proof}.
        
        \subsection{Expected Horizon lemma}\label{section:proofs_appendix:horizon}

The following lemma is a restatement of Lemma~\ref{diff_cont_lemma}.
\begin{lemma}\label{restated_diff_cont_lemma}
The expected total contribution to the difficulty when $\auxmdp$ terminates is approximately equal to~$H$.
Formally, it holds that:
\begin{displaymath}
H - \dmax - 1 \leq \E{\sum\limits_{t=1}^{\termaux(H)} D_t} \leq H + \dmax \, \, .
\end{displaymath}
\end{lemma}

\begin{proof}
We start by developing the expression of the probability the process continues after ~$T \in \N$ steps conditioned on the difficulty contribution accumulated until $T$ is $z \in \R^+$:
\begin{multline}\label{prob_t_less_than_term}
\Pr{T < \termaux(H) \middle\vert \sum\limits_{t=1}^T D_t=z} =\\
= \Pr{X_1 = X_2 = ... = X_T = 0 \middle\vert \sum\limits_{t=1}^T D_t=z} =\\
= \E{\Pr{X_1 = ... = X_T = 0 \middle\vert D_1, D_2, ..., D_{t-1}, D_t = z - \sum\limits_{t=1}^{T-1} D_t}} =\\
= \mathbb{E} \left[ \Pr{X_1 = 0 \middle\vert D_1} \cdot \Pr{X_2 = 0 \middle\vert D_2} \cdot ... \right. \\
\left. ... \cdot \Pr{x_{T-1} = 0\middle\vert D_{T-1}} \cdot \Pr{x_T=0 \middle\vert D_t = z - \sum\limits_{t=1}^{T-1} D_t} \right] =\\
= \E{\left( 1 - \frac{1}{H} \right)^{D_1} \cdot ... \cdot \left( 1 - \frac{1}{H} \right)^{D_{T-1}} \cdot \left( 1 - \frac{1}{H} \right)^{z - \sum\limits_{t=1}^{T-1} D_t}} =\\
=\E{\left( 1 - \frac{1}{H} \right)^z} = \left( 1 - \frac{1}{H} \right)^z \, \, .
\end{multline}
As for any~$t$,~$D_t\geq0$, the difficulty contribution up until the last step before termination is more than some $y \in \R^+$:
\begin{displaymath}
\sum\limits_{T=1}^{\termaux(H)-1} D_t \geq y
\end{displaymath}
happens iff there is some $T < \termaux(H)$ for which immediately after step $T$ the difficulty contribution surpasses~$y$:
\begin{displaymath}
\exists_{T \in \N}: \sum\limits_{t=1}^{T-1} D_t < y \leq \sum\limits_{t=1}^T D_t \wedge T <\termaux(H) \, \, .
\end{displaymath}
This gives:
\begin{multline*}
\Pr{\sum\limits_{T=1}^{\termaux(H)-1} D_t \geq y} =\\
= \Pr{\exists_{T \in \N}: \sum\limits_{t=1}^{T-1} D_t < y \leq \sum\limits_{t=1}^T D_t \wedge T <\termaux(H)}  \, \, .
\end{multline*}
By summing over all options of $T \in \N$ and then using the law of total expectation, we get:
\begin{multline*}
\Pr{\exists_{T \in \N}: \sum\limits_{t=1}^{T-1} D_t < y \leq \sum\limits_{t=1}^T D_t \wedge T <\termaux(H)} =\\
= \sum\limits_{T=1}^\infty\Pr{\sum\limits_{t=1}^{T-1} D_t < y \leq \sum\limits_{t=1}^T D_t \wedge T <\termaux(H)} =\\
= \sum\limits_{T=1}^\infty \E{\Pr{\sum\limits_{t=1}^{T-1} D_t < y \leq \sum\limits_{t=1}^T D_t \wedge T <\termaux(H) \middle\vert \sum\limits_{t=1}^T D_t = z}} \, \, .
\end{multline*}
Once $\sum\limits_{t=1}^T D_t = z$ is determined, $\termaux(H)$ only depends on $z$ and is independent of $\sum\limits_{t=1}^{T-1} D_t$.
Therefore:
\begin{multline*}
\Pr{\sum\limits_{T=1}^{\termaux(H)-1} D_t \geq y} =\\
=\sum\limits_{T=1}^\infty \E{\Pr{\sum\limits_{t=1}^{T-1} D_t < y \leq \sum\limits_{t=1}^T D_t \wedge T <\termaux(H) \middle\vert \sum\limits_{t=1}^T D_t = z}}=\\
= \sum\limits_{T=1}^\infty \mathbb{E} \left[ \Pr{\sum\limits_{t=1}^{T-1} D_t < y \leq \sum\limits_{t=1}^T D_t \middle\vert \sum\limits_{t=1}^T D_t = z} \cdot \right.\\
\left. \cdot \Pr{T <\termaux(H) \middle\vert \sum\limits_{t=1}^T D_t = z} \right] \, \, .
\end{multline*}
Now by substituting \eqref{prob_t_less_than_term} into this, we get:
\begin{multline}\label{expected_over_y_summed_by_t}
\Pr{\sum\limits_{T=1}^{\termaux(H)-1} D_t \geq y} =\\
= \sum\limits_{T=1}^\infty \E{\Pr{\sum\limits_{t=1}^{T-1} D_t < y \leq \sum\limits_{t=1}^T D_t \middle\vert \sum\limits_{t=1}^T D_t = z}  \cdot \left( 1 - \frac{1}{H} \right)^z} \, \, .
\end{multline}
Now, thanks to Assumption~\ref{assump_d_t_bounded}, we get:
\begin{displaymath}
\sum\limits_{t=1}^{T-1} D_t < y \leq z = \sum\limits_{t=1}^T D_t = \sum\limits_{t=1}^{T-1} D_t + \dmax < y + \dmax \, \, ,
\end{displaymath}
\begin{displaymath}
y \leq z < y + \dmax \, \, ,
\end{displaymath}
and:
\begin{equation}\label{bound_prob_end_at_z}
\left( 1 - \frac{1}{H} \right)^{y + \dmax} < \left( 1 - \frac{1}{H} \right)^z \leq \left( 1 - \frac{1}{H} \right)^y\, \, .
\end{equation}
Putting this back in \eqref{expected_over_y_summed_by_t}:
\begin{multline*}
\sum\limits_{T=1}^\infty \E{\Pr{\sum\limits_{t=1}^{T-1} D_t < y \leq \sum\limits_{t=1}^T D_t \middle\vert \sum\limits_{t=1}^T D_t = z} \cdot \left( 1 - \frac{1}{H} \right)^z} >\\
> \sum\limits_{T=1}^\infty \E{\Pr{\sum\limits_{t=1}^{T-1} D_t < y \leq \sum\limits_{t=1}^T D_t \middle\vert \sum\limits_{t=1}^T D_t = z} \cdot \left( 1 - \frac{1}{H} \right)^{y + \dmax}} =\\
= \sum\limits_{T=1}^\infty \E{\Pr{\sum\limits_{t=1}^{T-1} D_t < y \leq \sum\limits_{t=1}^T D_t \middle\vert \sum\limits_{t=1}^T D_t = z}} \cdot \left( 1 - \frac{1}{H} \right)^{y + \dmax} =\\
= \sum\limits_{T=1}^\infty \Pr{\sum\limits_{t=1}^{T-1} D_t < y \leq \sum\limits_{t=1}^T D_t} \cdot \left( 1 - \frac{1}{H} \right)^{y + \dmax} =\\
= \underbrace{\Pr{\exists_{T \in \N}: \sum\limits_{t=1}^T D_t \geq y}}_1 \cdot \left( 1 - \frac{1}{H} \right)^{y + \dmax} = \left( 1 - \frac{1}{H} \right)^{y + \dmax}
\end{multline*}
The probability is equal to 1 since if the game goes on forever the difficulty contributions will accumulate enough to pass $y$ thanks to assumption \ref{assump_d_t_not_zero}.
Overall, by using all the equations and inequalities above:
\begin{equation}\label{exp_horizon_lower_bound_prob}
\left( 1 - \frac{1}{H} \right)^{y + \dmax} < \Pr{\sum\limits_{T=1}^{\termaux(H)-1} D_t \geq y} \, \, .
\end{equation}
By using the upper bound from \eqref{bound_prob_end_at_z} similarly, it holds that:
\begin{equation}\label{exp_horizon_upper_bound_prob}
\Pr{\sum\limits_{T=1}^{\termaux(H)-1} D_t \geq y} \leq \left( 1 - \frac{1}{H} \right)^y \, \, .
\end{equation}
Now, calculating the expected difficulty contribution until $\termaux(H)-1$ and bounding from above using~\eqref{exp_horizon_upper_bound_prob}, we get:
\begin{multline*}
\E{\sum\limits_{T=1}^{\termaux(H)-1} D_t} = \int\limits_0^\infty \Pr{\sum\limits_{T=1}^{\termaux(H)-1} D_t \geq y} dy \leq\\
\leq \int\limits_0^\infty \left( 1 - \frac{1}{H} \right)^y dy = \left[ \frac{1}{\ln \left( 1 - \frac{1}{H} \right)} \left( 1 - \frac{1}{H} \right)^y dy \right]_0^\infty =
\end{multline*}
\begin{equation}\label{upper_bound_sum_term_minus_1}
= -\frac{1}{\ln \left( 1 - \frac{1}{H} \right)} = \frac{H}{-H \cdot \ln \left( 1 - \frac{1}{H} \right)} = \frac{H}{\ln \left( \left( 1 - \frac{1}{H}\right)^{-H} \right)} \leq\\
\leq\frac{H}{\ln \mathrm{e}} = H \, \, .
\end{equation}
And bounding from below using~\eqref{exp_horizon_lower_bound_prob} as well, we get:
\begin{multline}\label{lower_bound_sum_term_minus_1}
\E{\sum\limits_{T=1}^{\termaux(H)-1} D_t} = \int\limits_0^\infty \Pr{\sum\limits_{T=1}^{\termaux(H)-1} D_t \geq y} dy \geq\\
\geq \int\limits_0^\infty \left( 1 - \frac{1}{H} \right)^{y + \dmax} dy\\
= \left( 1 - \frac{1}{H} \right)^\dmax \cdot \left[ \frac{1}{\ln \left( 1 - \frac{1}{H} \right)} \left( 1 - \frac{1}{H} \right)^y dy \right]_0^\infty =\\
= -\frac{\left( 1 - \frac{1}{H} \right)^\dmax}{\ln \left( 1 - \frac{1}{H} \right)} = \frac{(H - 1) \left( 1 - \frac{1}{H} \right)^\dmax}{-(H - 1) \cdot \ln \left( 1 - \frac{1}{H} \right)} \geq \\
\geq \frac{(H - 1) \left( 1 - \frac{\dmax}{H} \right)}{\ln \left( \left( 1 - \frac{1}{H}\right)^{-H+1} \right)} \geq \frac{H - 1 - \dmax + \frac{\dmax}{H}}{\ln \mathrm{e}} \geq H - \dmax - 1 \, \, .
\end{multline}
Now, trivially following by Assumption~\ref{assump_d_t_bounded}, notice that:
\begin{displaymath}
\sum\limits_{T=1}^{\termaux(H)-1} D_t \leq \sum\limits_{T=1}^{\termaux(H)} D_t < \sum\limits_{T=1}^{\termaux(H)-1} D_t + \dmax \, \, .
\end{displaymath}
We now use the equation above with the bounds previously found \eqref{upper_bound_sum_term_minus_1} and \eqref{lower_bound_sum_term_minus_1} to bound the expectation from below:
\begin{displaymath}
\E{\sum\limits_{T=1}^{\termaux(H)} D_t} \geq \E{\sum\limits_{T=1}^{\termaux(H)-1} D_t} \geq H - \dmax - 1 \, \, ,
\end{displaymath}
and above:
\begin{displaymath}
\E{\sum\limits_{T=1}^{\termaux(H)} D_t} \leq \E{\sum\limits_{T=1}^{\termaux(H)-1} D_t} + \dmax \leq H + \dmax \, \, .
\end{displaymath}
Combining the previous two inequalities we get:
\begin{equation*}
 H - \dmax - 1 \leq \E{\sum\limits_{T=1}^{\termaux(H)} D_t} \leq H + \dmax \, \, . \qedhere
\end{equation*} 
\end{proof}

        \subsection{Bounding the Average Difference}\label{section:proofs_appendix:auxiliary_lemmas}
        
We first state some important classical results for Markov chains~\cite{serfozo2009basics,funderlic1986sensitivity}.

\begin{lemma}\label{thm54}
\cite{serfozo2009basics}
An irreducible Markov Chain has a positive distribution if and only all of its states are positive recurrent.
In that case, the stationary distribution is unique and has the following form:
\begin{displaymath}
\mu_i = \frac{1}{\E{\tau_i}} \, \, .
\end{displaymath}
\end{lemma}
        
The following lemma gives a way to calculate the expected cumulative sum of a random variable which depends on the current state of the chain until some chosen state is entered.
\begin{lemma}\label{prop69}
\cite{serfozo2009basics}
Let~$Y_n$ be an irreducible positive recurrent Markov chain with stationary distribution $\mu$.
Suppose~$V_n$,~$n \geq 1$, are real-valued random variables associated with the chain such that\\
\begin{displaymath}
\E{V_n|Y_1,Y_2,...,Y_n} = a_{Y_n}, \ \ \ n \geq 1 \, \, ,
\end{displaymath}
where $a_j$ are constants.
Then, for the hitting time $\tau_i$ of a fixed state $i$, it holds that
\begin{displaymath}
\E{\sum\limits_{n=1}^{\tau_i - 1} V_n \middle\vert T_1 = i} = \frac{1}{\mu_i} \sum\limits_{j \in \fancyS} a_j \mu_j \, \, ,
\end{displaymath}
provided the last sum is absolutely convergent.
\end{lemma}

\begin{lemma}\label{cor79}
\cite{serfozo2009basics}
For a fixed integer $\ell$, the process $\tilde{Y}_n = (Y_n,...,Y_{n+\ell})$ is an ergodic Markov chain on $\fancyS^{\ell + 1}$ with stationary distribution
\begin{displaymath}
\mu(i) = \mu_{i_0} p_{i_0,i_1} \cdot \cdot \cdot p_{i_{\ell - 1}, i_\ell} \, \, .
\end{displaymath}
Hence, for $f:\fancyS^{\ell + 1} \to \R$,
\begin{displaymath}
\lim\limits_{n \to \infty} n^{-1} \sum\limits_{m=1}^n f(\tilde{Y}_m) = \sum\limits_{i \in \fancyS^{\ell + 1}} f(i) \mu(i) \ \ \ \text{a.s.},
\end{displaymath}
provided the sum is absolutely convergent.
\end{lemma}
This lemma is a generalization of Lemma \ref{prop69}.
First, instead of a random variable which depends on a single state, it allows using any (deterministic) function of the last $(\ell + 1)$ states.
This is a powerful notion since it depends on multiple states instead of one.
Second, instead of a cumulative sum until entering a said state, it provides the average of the function when the process runs indefinitely.

\begin{lemma}\label{distance_stationary_distribution}
\cite{funderlic1986sensitivity}
Let Y, and $\tilde{Y}$ be ergodic Markov chains with transition matrices $P$ and $\tilde{P}$ and stationary distributions $\mu$ and $\tilde{\mu}$ then:
\begin{displaymath}
\norminf{\mu - \tilde{\mu}} = \bigO{\norminf{P - \tilde{P}}} \, \, .
\end{displaymath}
\end{lemma}


We first restate the transition probabilities in $\auxmdp$.
\begin{lemma}\label{mdp_prime_transition_p}
The transition probability of $\auxmdp$ is:
\begin{displaymath}
\Paux(i,j) =
\begin{cases}
    \left( 1 - \frac{1}{H} \right)^{\Dmdp(i,j)} \Porg(i,j) + 1 - \left( 1 - \frac{1}{H} \right)^{\Dmdp(i,j)} & j = \sinitaux \\
    \left( 1 - \frac{1}{H} \right)^{\Dmdp(i,j)} \Porg(i,j) & \text{o.w}
\end{cases} \, \, .
\end{displaymath}
\end{lemma}
\begin{proof}
Follows immediately from the definition of $\auxmdp$.
\end{proof}

We now use the Lemma~\ref{distance_stationary_distribution} to obtain a bound for the difference between the stationary distributions of the 2 MDPs.
\begin{lemma}\label{lemma_stationary_dist_equal}
The stationary distributions of $\orgmdp$ and $\auxmdp$ are equal up to $\bigO{\frac{1}{H}}$.
Formally:
\begin{displaymath}
\norminf{\piorg - \piaux} = \bigO{\frac{1}{H}}
\end{displaymath}
\end{lemma}
\begin{proof}
Using lemma \ref{distance_stationary_distribution}, and the fact that the MDPs are ergodic (Lemma \ref{erogidic_thm}) we get that:
\begin{equation}\label{upper_bound_with_norm_matrices_diff}
\norminf{\piorg - \piaux} = \bigO{\norminf{\Porg - \Paux}} \, \, .
\end{equation}
Using the definition of the infinity norm for matrices~-- the maximum column sum the absolute values of elements, we get:
\begin{equation}\label{upper_bound_with_max_diff}
\norminf{\Porg - \Paux} \leq \abs{\fancyS} \max\limits_{i, j \in \fancyS} \abs{\Porg(i,j) - \Paux(i,j)} \, \, .
\end{equation}
We will now prove that the last term is $\bigO{\frac{1}{H}}$.
There are 2 cases from lemma \ref{mdp_prime_transition_p}.
For the first case $j=\sinitaux$ we get:
\begin{multline*}
\abs{\Porg(i,j) - \Paux(i,j)} =\\
= \abs{\Porg(i,j) - \left( 1 - \frac{1}{H} \right)^{\Dmdp(i,j)} \Porg(i,j) - 1 + \left( 1 - \frac{1}{H} \right)^{\Dmdp(i,j)}} =\\
= \abs{\left( \Porg(i,j) - 1 \right) \left( 1 - \left( 1 - \frac{1}{H} \right)^{\Dmdp(i,j)} \right)} \leq\\
\leq \abs{\Porg(i,j) - 1} \cdot \abs{1 - \left( 1 - \frac{1}{H} \right)^\dmax} \, \, .
\end{multline*}
Therefore for $H \gg \dmax$ it holds that:
\begin{multline}\label{eq_P_diff_sinit}
\abs{\Porg(i,j) - \Paux(i,j)} \leq \abs{\Porg(i,j) - 1} \cdot \abs{1 - \left( 1 - \frac{\dmax}{H} \right)} \leq\\
\leq \abs{\Porg(i,j) - 1} \cdot \frac{\dmax}{H} = \bigO{\frac{1}{H}} \, \, .
\end{multline}
For the other case, we get similarly that $j \not= \sinitaux$:
\begin{multline}\label{eq_P_diff_not_sinit}
\abs{\Porg(i,j) - \Paux(i,j)} = \abs{\Porg(i,j) - \left( 1 - \frac{1}{H} \right)^{\Dmdp(i,j)} \Porg(i,j)}\\
= \abs{\Porg(i,j) - \left( 1 - \frac{1}{H} \right)^{\Dmdp(i,j)} \Porg(i,j)} = \abs{\Porg(i,j)} \abs{1 - \left( 1 - \frac{1}{H} \right)^{\Dmdp(i,j)}}=\\
 = \bigO{\frac{1}{H}} \, \, .
\end{multline}
By plugging \eqref{eq_P_diff_sinit} and \eqref{eq_P_diff_not_sinit} in \eqref{upper_bound_with_max_diff} we get that:
\begin{displaymath}
\norminf{\Porg - \Paux} \leq \abs{\fancyS} \max\limits_{i, j \in \fancyS} \abs{\Porg(i,j) - \Paux(i,j)} = \abs{\fancyS} \bigO{\frac{1}{H}} = \bigO{\frac{1}{H}} \, \, .
\end{displaymath}
Plugging this in \eqref{upper_bound_with_norm_matrices_diff} to obtain:
\begin{align*}
\norminf{\piorg - \piaux} = \bigO{\norminf{\Porg - \Paux}} = \bigO{\frac{1}{H}} \, \, . &\qedhere
\end{align*}
\end{proof}

Next, we have 2 similar lemmas which bound the difference between the rewards and difficulty contributions in every state in both MDPs.
\begin{lemma}\label{lemma_rewards_equal}
The rewards in every step of $\orgmdp$ and $\auxmdp$ are equal up to $\bigO{\frac{1}{H}}$.
Formally:
\begin{displaymath}
\norminf{\hat{R} - \hat{R}'_H} = \bigO{\frac{1}{H}} \, \, .
\end{displaymath}
\end{lemma}
\begin{proof}
Remember the definitions of~$\hat{R}$ and~$\hat{R}'_H$:
\begin{displaymath}
\hat{R}(i) = \sum\limits_{j \in \fancyS} \Rmdp(i,j)  \Porg(i,j) \, \, ,
\end{displaymath}
and
\begin{displaymath}
\hat{R}'_H(i) = \sum\limits_{j \in \fancyS} \Rmdp(i,j)  \Paux(i,j) \, \, .
\end{displaymath}
We use these to develop our expression and obtain:
\begin{multline*}
\norminf{\hat{R} - \hat{R}'_H} = \max_{i \in \fancyS} \abs{\sum\limits_{j \in \fancyS} \Rmdp(i,j)  \Porg(i,j) - \sum\limits_{j \in \fancyS} \Rmdp(i,j)  \Paux(i,j)} =\\
= \max_{i \in \fancyS} \abs{\sum\limits_{j \in \fancyS} \Rmdp(i,j) \left( \Porg(i,j) - \Paux(i,j) \right)} \, \, \, .
\end{multline*}
By using Assumption~\ref{assump_r_t_bounded} and the triangle inequality we get that:
\begin{equation}\label{reward_equal_lemma_eq}
\norminf{\hat{R} - \hat{R}'_H} \leq \rmax \cdot \max_{i \in \fancyS} \sum\limits_{j \in \fancyS} \abs{\Porg(i,j) - \Paux(i,j)} \,\, .
\end{equation}
By using equations \eqref{eq_P_diff_sinit} and \eqref{eq_P_diff_not_sinit} inside the proof of Lemma~\ref{lemma_stationary_dist_equal}, we see that:
\begin{displaymath}
\abs{\Porg(i,j) - \Paux(i,j)} = \bigO{\frac{1}{H}} \, \, .
\end{displaymath}
We use this to bound the expression in \eqref{reward_equal_lemma_eq}, and obtain:
\begin{align*}
\norminf{\hat{R} - \hat{R}'_H} \leq \rmax \cdot \max_{i \in \fancyS} \abs{\fancyS} \cdot \bigO{\frac{1}{H}} = \bigO{\frac{1}{H}} \, \, .& \qedhere
\end{align*}
\end{proof}

\begin{lemma}\label{lemma_difficulty_cont_equal}
The difficulty contributions in every step of $\orgmdp$ and $\auxmdp$ are equal up to $\bigO{\frac{1}{H}}$.
Formally:
\begin{displaymath}
\norminf{\hat{D} - \hat{D}'_H} = \bigO{\frac{1}{H}} \, \, .
\end{displaymath}
\end{lemma}
\begin{proof}
Same proof as lemma \ref{lemma_rewards_equal} but with $\hat{D}$ and $\hat{D}'_H$ instead of $\hat{R}$ and $\hat{R}'_H$.
\end{proof}


Now, we will use the previous 2 lemmas to prove Lemma~\ref{lemma_expected_reward_diff_cont_equal}.
We first restate the lemma.
\begin{lemma}\label{restated_lemma_expected_reward_diff_cont_equal}
The expected reward per step of $\orgmdp$ and $\auxmdp$ are equal up to $\bigO{\frac{1}{H}}$.
Formally, it holds that:
\begin{displaymath}
\abs{\inner{\hat{R}}{\piorg} - \inner{\hat{R}'_H}{\piaux}} = \bigO{\frac{1}{H}} \, \, .
\end{displaymath}
Furthermore, The expected difficulty contribution per step of $\orgmdp$ and $\auxmdp$ are also equal up to $\bigO{\frac{1}{H}}$.
Formally, it holds that:
\begin{displaymath}
\abs{\inner{\hat{D}}{\piorg} - \inner{\hat{D}'_H}{\piaux}} = \bigO{\frac{1}{H}} \, \, .
\end{displaymath}
\end{lemma}
\begin{proof}
We begin by developing the the difference in the expected reward between $\orgmdp$ and $\auxmdp$ to obtain:
\begin{multline*}
\abs{\inner{\hat{R}}{\piorg} - \inner{\hat{R}'_H}{\piaux}} = \abs{\inner{\hat{R}}{\piorg} - \inner{\hat{R}}{\piaux} + \inner{\hat{R}}{\piaux} + \inner{\hat{R}'_H}{\piaux}} =\\
= \abs{\inner{\hat{R}}{\piorg - \piaux} - \inner{\hat{R} - \hat{R}'_H}{\piaux}} \, \, .
\end{multline*}
By using the triangle inequality and Cauchy-Schwarz inequality, we get that:
\begin{multline*}
\abs{\inner{\hat{R}}{\piorg - \piaux} - \inner{\hat{R} - \hat{R}'_H}{\piaux}} \leq \abs{\inner{\hat{R}}{\piorg - \piaux}} + \abs{\inner{\hat{R} - \hat{R}'_H}{\piaux}} \leq \\
\leq \norm{\hat{R}}_2 \cdot \norm{\piorg - \piaux}_2 + \norm{\hat{R} - \hat{R}'_H}_2 \cdot \, \, . \norm{\piaux}_2
\end{multline*}
Under a finite state space, the $L^2$ norm and $L^\infty$ are asymptotically equivalent so for any~$x$ it holds that:
\begin{displaymath}
\norm{\vec{x}}_2 = \bigO{\norminf{\vec{x}}} \, \, .
\end{displaymath}
We use this to obtain:
\begin{multline*}
\abs{\inner{\hat{R}}{\piorg} - \inner{\hat{R}'_H}{\piaux}} \leq \norm{\hat{R}}_2 \cdot \norm{\piorg - \piaux}_2 + \norm{\hat{R} - \hat{R}'_H}_2 \cdot \norm{\piaux}_2 \leq\\
\leq \bigO{\norminf{\hat{R}} \cdot \norminf{\piorg - \piaux} + \norminf{\hat{R} - \hat{R}'_H} \cdot \norminf{\piaux}} \, \, .
\end{multline*}
Then, by using Assumption \ref{assump_r_t_bounded}, the fact that a stationary distribution is comprised of probabilities ($\leq 1$) and Lemmas \ref{lemma_rewards_equal} and \ref{lemma_difficulty_cont_equal}, we get:
\begin{multline*}
\bigO{\norminf{\hat{R}} \cdot \norminf{\piorg - \piaux} + \norminf{\hat{R} - \hat{R}'_H} \cdot \norminf{\piaux}} \leq\\
\leq \bigO{\rmax \cdot \bigO{\frac{1}{H}} + \bigO{\frac{1}{H}} \cdot 1} = \bigO{\frac{1}{H}} \, \, .
\end{multline*}
We combine all the previous inequalities and get that:
\begin{displaymath}
\abs{\inner{\hat{R}}{\piorg} - \inner{\hat{R}'_H}{\piaux}} = \bigO{\frac{1}{H}} \, \, .
\end{displaymath}
The exact with $\hat{D}$ and $\hat{D}'_H$ instead of $\hat{R}$ and $\hat{R}'_H$ gives the second result.
\end{proof}


Now we restate and prove Lemma~\ref{lemma_d_t_org_aux_not_zero}.
\begin{lemma}\label{restated_lemma_d_t_org_aux_not_zero}
The average difficulty contribution per step in $\orgmdp$ and $\auxmdp$ is more than some constant $\varepsilon$ > 0.
Formally, it holds that
\begin{displaymath}
\inner{\hat{D}}{\piorg} > \varepsilon \, \, ,
\end{displaymath}
and
\begin{displaymath}
\inner{\hat{D}'_H}{\piaux} > \varepsilon \, \, .
\end{displaymath}
\end{lemma}
\begin{proof}
The first part of this proof is a directly corollary of the ergodicity and Assumption~\ref{assump_d_t_not_zero}.

To prove the second part of the lemma we use the second part of Lemma~ \ref{restated_lemma_expected_reward_diff_cont_equal} and the first part of this lemma to get:
\begin{displaymath}
\inner{\hat{D}'_H}{\piorg} > \varepsilon - \bigO{\frac{1}{H}} \, \, .
\end{displaymath}
For~$H \gg \dmax$, there is some constant~$\varepsilon' > 0$ such that:
\begin{displaymath}
\inner{\hat{D}'_H}{\piorg} > \varepsilon' \, \, .
\end{displaymath}
To ease notation, since all we care about is that there is a lower bound we redefine $\varepsilon$ to be minimum of the previous value of $\varepsilon$ and the value of $\varepsilon$'.
\end{proof}

        \subsection{Simplifying the Revenues}\label{section:proofs_appendix:simplify}
        
In the following we restate and prove Lemmas~\ref{simplified_revorg} and~\ref{simplified_revaux}.

\begin{lemma}\label{restated_simplified_revorg}
The revenue in $\orgmdp$ is equal to the average expected reward per step divided by the average expected difficulty contribution per step.
Formally, it holds that:
\begin{displaymath}
\revorg = \frac{\inner{\hat{R}}{\piorg}}{\inner{\hat{D}}{\piorg}} \, \, .
\end{displaymath}
\end{lemma}
\begin{proof}
Recall from the definition $\revorg$ that:
\begin{displaymath}
\revorg = \E{\lim\limits_{T\to\infty} \frac{\sum\limits_{t=1}^T R_t}{\sum\limits_{t=1}^T D_t}} \, \, .
\end{displaymath}
We first analyze the expression within the expectation and get:
\begin{equation}\label{rev_arr_1_t}
\lim\limits_{T\to\infty} \frac{\sum\limits_{t=1}^T R_t}{\sum\limits_{t=1}^T D_t} =  \frac{\lim\limits_{T\to\infty} \frac{1}{T} \sum\limits_{t=1}^T R_t}{\lim\limits_{T\to\infty} \frac{1}{T} \sum\limits_{t=1}^T D_t} \, \, .
\end{equation}
This equality is true only if both limits are well defined and the denominator is not 0.
We will see that this is indeed the case after fully developing this expression.

Now, using Lemma~\ref{cor79} with $\ell = 1$ when choosing:
\begin{displaymath}
f: (X_t, X_{t+1}) \mapsto \Rmdp(X_t,X_{t+1})\, \, ,
\end{displaymath}
for the nominator and:
\begin{displaymath}
f: (X_t, X_{t+1}) \mapsto \Dmdp(X_t,X_{t+1})\, \, ,
\end{displaymath}
for the denominator, we get that:
\begin{displaymath}
\lim\limits_{T\to\infty} \frac{1}{T} \sum\limits_{t=1}^T R_t = \sum\limits_{(i,j) \in \fancyS^2} \Rmdp(i,j) \piorg(i) \Porg(i,j)\, \, ,
\end{displaymath}
and:
\begin{displaymath}
\lim\limits_{T\to\infty} \frac{1}{T} \sum\limits_{t=1}^T D_t = \sum\limits_{(i,j) \in \fancyS^2} \Dmdp(i,j) \piorg(i) \Porg(i,j)\, \, .
\end{displaymath}
Then after substituting the above terms in \eqref{rev_arr_1_t} and then taking the sums apart, we get:
\begin{multline*}
\frac{\lim\limits_{T\to\infty} \frac{1}{T} \sum\limits_{t=1}^T R_t}{\lim\limits_{T\to\infty} \frac{1}{T} \sum\limits_{t=1}^T D_t} = \frac{\sum\limits_{(i,j) \in \fancyS^2} \Rmdp(i,j) \piorg(i) \Porg(i,j)}{\sum\limits_{(i,j) \in \fancyS^2} \Dmdp(i,j) \piorg(i) \Porg(i,j)}=\\
 = \frac{\sum\limits_{i \in s} \sum\limits_{j \in s} \Rmdp(i,j) \piorg(i) \Porg(i,j)}{\sum\limits_{i \in s} \sum\limits_{j \in s} \Dmdp(i,j) \piorg(i) \Porg(i,j)} = \frac{\sum\limits_{i \in s} \hat{R}(i) \piorg(i)}{\sum\limits_{i \in s} \hat{D}(i) \piorg(i)}  \, \, . 
\end{multline*}
Then, by combining all the previous equations, we get that:
\begin{displaymath}
\revorg = \E{\lim\limits_{T\to\infty} \frac{\sum\limits_{t=1}^T R_t}{\sum\limits_{t=1}^T D_t}} = \E{\frac{\sum\limits_{i \in s} \hat{R}(i) \piorg(i)}{\sum\limits_{i \in s} \hat{D}(i) \piorg(i)}} \, \, .
\end{displaymath}
Since the expression in the last expectation is constant, its expectation is itself.
We use that to obtain:
\begin{align*}
\revorg = \frac{\sum\limits_{i \in s} \hat{R}(i) \piorg(i)}{\sum\limits_{i \in s} \hat{D}(i) \piorg(i)} = \frac{\inner{\hat{R}}{\piorg}}{\inner{\hat{D}}{\piorg}} \, \, . &\qedhere
\end{align*}
\end{proof}


\begin{lemma}\label{restated_simplified_revaux}
The revenue in $\auxmdp$ is equal to the average expected reward per step divided by the average expected difficulty contribution per step up to $\bigO{\frac{1}{H}}$.
Formally, it holds that:
\begin{displaymath}
\abs{\revaux - \frac{\inner{\hat{R}'}{\piaux}}{\inner{\hat{D}'}{\piaux}}} = \bigO{\frac{1}{H}} \, \, .
\end{displaymath}
\end{lemma}
\begin{proof}
From the definition of $\revaux$ and from the lower bound in Lemma~\ref{restated_diff_cont_lemma}, we lower bound the fraction in the expression:
\begin{multline*}
\revaux = \frac{1}{H}\E{\sum\limits_{t=1}^{\termaux(H)} R_t}
\geq \frac{\E{\sum\limits_{t=1}^{\termaux(H)} R_t}}{{\E{\sum\limits_{t=1}^{\termaux(H)} D_t}} + \dmax + 1} \, \, .
\end{multline*}
Note that:
\begin{enumerate}
\item $(\termaux(H) + 1)$ is the hitting time of $\sinitaux$ and
\item in $\auxmdp$, $\E{R_t | X_t = i} = \hat{R}'_H(i)$ and $\E{D_t | X_t = i} = \hat{D}'_H(i)$.
\end{enumerate}
Therefore Lemma~\ref{prop69} can be used twice (once in the nominator and once in the denominator) to obtain:
\begin{multline}\label{lower_bound_rev_aux_middle}
\revaux \geq \frac{\E{\sum\limits_{t=1}^{\termaux(H)} R_t}}{{\E{\sum\limits_{t=1}^{\termaux(H)} D_t}} + \dmax + 1} =\\
= \frac{(\piaux(\sinitaux))^{-1} \sum\limits_{i \in \fancyS} \hat{R}'_H(i) \piaux(i)}{(\piaux(\sinitaux))^{-1} \sum\limits_{i \in \fancyS} \hat{D}'_H(i) \piaux(i) + \dmax + 1} =\\
 = \frac{\inner{\hat{R}'}{\piaux}}{\inner{\hat{D}'}{\piaux} + (\dmax + 1) \piaux(\sinitaux)} \, \, .
\end{multline}
By using the upper bound in Lemma~\ref{restated_diff_cont_lemma}, and then continuing similarly, we get:
\begin{multline}\label{upper_bound_rev_aux_middle}
\revaux = \frac{1}{H}\E{\sum\limits_{t=1}^{\termaux(H)} R_t}
\leq \frac{\E{\sum\limits_{t=1}^{\termaux(H)} R_t}}{{\E{\sum\limits_{t=1}^{\termaux(H)} D_t}} - \dmax} =\\
= \frac{(\piaux(\sinitaux))^{-1} \sum\limits_{i \in \fancyS} \hat{R}'_H(i) \piaux(i)}{(\piaux(\sinitaux))^{-1} \sum\limits_{i \in \fancyS} \hat{D}'_H(i) \piaux(i) - \dmax} =\\
=\frac{\inner{\hat{R}'}{\piaux}}{\inner{\hat{D}'}{\piaux} - \dmax \piaux(\sinitaux)}
\end{multline}
As stated before, $(\termaux(H) + 1)$ is the hitting time of $\sinitaux$.
So, thanks to Lemma~\ref{thm54}, it holds that:
\begin{displaymath}
\piaux(\sinitaux) = \frac{1}{\E{\termaux(H) + 1}} \, \, .
\end{displaymath}
By using Assumption~\ref{assump_d_t_bounded} and Lemma \ref{restated_diff_cont_lemma}, we get:
\begin{multline*}
\dmax \cdot \E{\termaux(H)} \geq \E{\sum\limits_{t=1}^{\termaux(H)} \dmax} \geq\\
\geq \E{\sum\limits_{t=1}^{\termaux(H)} D_t} \geq H - \dmax - 1  \, \, .
\end{multline*}
We proceed to lower bound $\piaux(\sinitaux)$. We get:
\begin{displaymath}
\E{\termaux(H)} \geq \frac{H - \dmax - 1}{\dmax} = \frac{H - 1}{\dmax} - 1 \, \, .
\end{displaymath}
and: 
\begin{displaymath}
\piaux(\sinitaux) \leq \frac{1}{\frac{H - 1}{\dmax}} = \frac{\dmax}{H - 1} \, \, .
\end{displaymath}
Plugging this into \eqref{lower_bound_rev_aux_middle} yields:
\begin{equation}\label{lower_bound_rev_aux}
\revaux \geq \frac{\inner{\hat{R}'}{\piaux}}{\inner{\hat{D}'}{\piaux} + (\dmax + 1) \piaux(\sinitaux)} \geq \frac{\inner{\hat{R}'}{\piaux}}{\inner{\hat{D}'}{\piaux} + \frac{\dmax^2 + \dmax}{H - 1}} \, \, .
\end{equation}
And, plugging the lower bound into \eqref{upper_bound_rev_aux_middle} yields:
\begin{equation}\label{upper_bound_rev_aux}
\revaux \leq \frac{\inner{\hat{R}'}{\piaux}}{\inner{\hat{D}'}{\piaux} - \dmax \piaux(\sinitaux)} \leq \frac{\inner{\hat{R}'}{\piaux}}{\inner{\hat{D}'}{\piaux} - \frac{\dmax^2}{H - 1}} \, \, .
\end{equation}
By using both bounds \eqref{lower_bound_rev_aux} and \eqref{upper_bound_rev_aux} and noticing that the desired ratio also lies within the bounds we get that:
\begin{multline*}
\abs{\revaux - \frac{\inner{\hat{R}'}{\piaux}}{\inner{\hat{D}'}{\piaux}}} \leq\\
\leq \abs{\frac{\inner{\hat{R}'}{\piaux}}{\inner{\hat{D}'}{\piaux} - \frac{\dmax^2}{H - 1}} - \frac{\inner{\hat{R}'}{\piaux}}{\inner{\hat{D}'}{\piaux} + \frac{\dmax^2 + \dmax}{H - 1}}} \leq\\
\leq \abs{\frac{\inner{\hat{R}'}{\piaux} \cdot \left( \inner{\hat{D}'}{\piaux} + \frac{\dmax^2 + \dmax}{H-1} \right) - \inner{\hat{R}'}{\piaux} \cdot \left( \inner{\hat{D}'}{\piaux} - \frac{\dmax^2}{H-1} \right)}{\left( \inner{\hat{D}'}{\piaux} + \frac{\dmax^2 + \dmax}{H-1} \right) \cdot \left( \inner{\hat{D}'}{\piaux} - \frac{\dmax^2}{H-1} \right)}} \leq\\
\leq \abs{\frac{\inner{\hat{R}'}{\piaux} \cdot \frac{2\dmax^2 + \dmax}{H-1}}{\left( \inner{\hat{D}'}{\piaux} + \frac{\dmax^2 + \dmax}{H-1} \right) \cdot \left( \inner{\hat{D}'}{\piaux} - \frac{\dmax^2}{H-1} \right)}}.
\end{multline*}
When taking into consideration $H \gg \dmax$ and using Lemma \ref{restated_lemma_d_t_org_aux_not_zero}, we get that:
\begin{multline*}
\abs{\frac{\inner{\hat{R}'}{\piaux} \cdot \frac{2\dmax^2 + \dmax}{H-1}}{\left( \inner{\hat{D}'}{\piaux} + \frac{\dmax^2 + \dmax}{H-1} \right) \cdot \left( \inner{\hat{D}'}{\piaux} - \frac{\dmax^2}{H-1} \right)}} < \\
< \varepsilon^{-2} \cdot {\inner{\hat{R}'}{\piaux}} \cdot \frac{2\dmax^2 + \dmax}{H-1}
\end{multline*}
Together with Assumption \ref{assump_r_t_bounded} and the property that a stationary distribution sums to 1, we obtain:
\begin{displaymath}
\varepsilon^{-2} \cdot {\inner{\hat{R}'}{\piaux}} \cdot \frac{2\dmax^2 + \dmax}{H-1} \leq \frac{3\dmax^2 \cdot \varepsilon^{-2} \cdot \rmax}{H-1} = \bigO{\frac{1}{H}} \, \, .
\end{displaymath}
We combine all the previous inequalities to finally derive:
\begin{align*}
\abs{\revaux - \frac{\inner{\hat{R}'}{\piaux}}{\inner{\hat{D}'}{\piaux}}} = \bigO{\frac{1}{H}} \, \, . &\qedhere
\end{align*}
\end{proof}

\end{document}

%% file: style.tex
\usepackage{amsmath} 
\usepackage{amsthm} 
\usepackage{amssymb} 
\usepackage{amsfonts} 
\usepackage{array}
\usepackage{tcolorbox}
\usepackage{caption}

\usepackage{epstopdf}
\usepackage{hyperref}
\usepackage{fancyhdr}
\usepackage{graphicx}
\usepackage{mathtools} 
\usepackage{url}
\usepackage{subcaption}
\usepackage{soul}
\usepackage{tabularx,environ}
\usepackage{xspace}
\usepackage{tikz}
\usetikzlibrary{positioning}
\usepackage{siunitx}

\renewcommand{\Pr}[1]{\text{Pr} \left( #1 \right)}
\newcommand{\E}[1]{\mathbb{E} \left[ #1 \right]}
\newcommand{\Epolicy}[1]{\mathbb{E}^\pi \left[ #1 \right]}

\newcommand{\abs}[1]{\left| {#1} \right|}
\newcommand{\norm}[1]{\left\lVert {#1} \right\rVert}
\newcommand{\norminf}[1]{\norm{#1}_\infty}
\newcommand\inner[2]{\langle {#1}, {#2} \rangle}
\newcommand{\bigO}[1]{O \left( {#1} \right)}

\newcommand{\fancyS}{\mathcal{S}}
\newcommand{\fancyA}{\mathcal{A}}
\newcommand{\R}{\mathbb{R}}
\newcommand{\N}{\mathbb{N}}

\newcommand{\rmax}{\ensuremath{{R_\text{max}}}}
\newcommand{\dmax}{\ensuremath{{D_\text{max}}}}
\newcommand{\rev}{\ensuremath{\text{REV}}}
\newcommand{\orgmdp}{\ensuremath{\text{ARR-MDP}}}
\newcommand{\auxmdp}{\ensuremath{\text{PT-MDP}}}
\newcommand{\revorg}{\ensuremath{\text{REV}_\text{ARR}}}
\newcommand{\revaux}{\ensuremath{\text{REV}_\text{PT}}}
\newcommand{\termaux}{\ensuremath{\textit{Term}}}
\newcommand{\Porg}{\ensuremath{P}}
\newcommand{\Paux}{\ensuremath{{P'_H}}}
\newcommand{\Dmdp}{\ensuremath{D}}
\newcommand{\Rmdp}{\ensuremath{R}}
\newcommand{\piorg}{\ensuremath{\mu}}
\newcommand{\piaux}{\ensuremath{\mu'_H}}
\newcommand{\sinit}{{\ensuremath{s_\text{init}}}}
\newcommand{\sinitaux}{{\ensuremath{s_\text{term}}}}

\newtheorem{theorem}{Theorem}
\newtheorem{corollary}[theorem]{Corollary}
\newtheorem{lemma}[theorem]{Lemma}
\newtheorem{assumption}{Assumption} 